\documentclass[11pt,a4paper]{article}
\usepackage{fullpage}

\usepackage{amssymb}
\usepackage{amsthm}
\usepackage{mathtools}
\usepackage{bm}
\usepackage{hyperref}
\usepackage[numbers,sort&compress]{natbib}
\usepackage[capitalise]{cleveref}

\newcommand{\cards}[1]{| #1 |}
\newcommand{\sset}[1]{\left\{ #1\right\}}
\newcommand{\fwhs}[1]{\; | \; #1 }
\renewcommand\vec{\bm}
\DeclareMathOperator{\sw}{SW}
\DeclareMathOperator{\opt}{OPT}
\DeclareMathOperator{\rev}{REV}
\DeclareMathOperator*{\argmax}{argmax}

\theoremstyle{definition}
\newtheorem{definition}{Definition}

\theoremstyle{plain}
\newtheorem{theorem}{Theorem}[section]
\newtheorem{lemma}[theorem]{Lemma}
\newtheorem{claim}[theorem]{Claim}
\newtheorem{corollary}[theorem]{Corollary}
\newtheorem{observation}{Observation}
\Crefname{claim}{Claim}{Claims}

\title{Nash Equilibria in Auctions with Pacing Strategies: \\Complexity and Inefficiency\thanks{Aris Filos-Ratsikas and Mohamad Latifian were supported by the UK Engineering and Physical Sciences Research Council (EPSRC) grant EP/Y003624/1. Charalampos Kokkalis was supported by an EPSRC DTA Scholarship (Reference EP/W524384/1).}}
\author{
  Aris Filos-Ratsikas\\
  Imperial College London, UK\\
  \href{mailto:aris.filos-ratsikas@imperial.ac.uk}{\small{\texttt{aris.filos-ratsikas@imperial.ac.uk}}}
  \and
  Charalampos Kokkalis\\
  University of Edinburgh, UK\\
  \href{mailto:charalampos.kokkalis@ed.ac.uk}{\small{\texttt{charalampos.kokkalis@ed.ac.uk}}}
  \and
  Mohamad Latifian\\
  University of Edinburgh, UK\\
  \href{mailto:mohamad.latifian@ed.ac.uk}{\small{\texttt{mohamad.latifian@ed.ac.uk}}}
}
\date{September 28, 2026}

\begin{document}

\maketitle

\begin{abstract}
We introduce and study \emph{Auctions with Pacing Strategies (APS) games}, a full-information model in which utility-maximizing bidders compete across many simultaneous first-price auctions, each choosing a single \emph{pacing multiplier} that uniformly scales their values into bids.
We settle three central questions.
First, we show that there are instances that admit no approximate pure Nash equilibria.
Then, we prove that the problem of deciding whether an APS game admits an (approximate) equilibrium is NP-complete in general, but can be solved in polynomial time if either the number of bidders or the number of items is fixed.
Finally, when an equilibrium does exist, we characterize its inefficiency exactly, showing that both the Price of Anarchy and the Price of Stability equal $\frac{e}{e-1}$. 
\end{abstract}

\section{Introduction}
Auctions are among the oldest and most widely used mechanisms for allocating scarce resources, dating back to antiquity and today underpinning markets as varied as the sale of art, treasury bonds, and spectrum licenses. Over the past two decades their most consequential application has arguably become \emph{online advertising}: each time a webpage loads, the advertising slots it displays are sold in real time through auctions run by ad exchanges, a process repeated billions of times a day. The scale is enormous: advertising accounts for the bulk of the income of the largest internet platforms, generating, for instance, roughly \$265 billion for Alphabet in 2024 alone, about three quarters of the company's total revenue.\footnote{Alphabet Inc., 2024 Annual Report (Form 10-K). See also \url{https://www.statista.com/statistics/266249/advertising-revenue-of-google/}.} For much of this history the format of choice was the \emph{second-price} auction and its generalizations, valued for being truthful; more recently, however, the industry has undergone a pronounced shift towards \emph{first-price} auctions, now the dominant format in display advertising exchanges. Unlike the second-price format, a first-price auction is not truthful: bidding one's true value leaves no surplus, so a bidder must strategically shade her bid, and the resulting outcome depends on how all participants reason about one another. This places the game-theoretic analysis of first-price auctions, and in particular the existence, computation, and efficiency of their equilibria, at the heart of understanding the modern advertising ecosystem.

How should a bidder decide what to bid? At the scale of online advertising, computing a bespoke bid for each auction is out of the question: a single advertiser may take part in millions of auctions per day, each cleared in milliseconds, so reasoning separately about the optimal bid in every one of them is infeasible. What advertisers do have is a reliable handle on their \emph{values}: much of the machinery of modern ad platforms is devoted to estimating how much a given impression is worth, through audience targeting and click-through or conversion prediction. This asymmetry, plentiful information about values but no room to strategize auction by auction, motivates a simple and by now standard heuristic~\citep{conitzer2022pacing}: rather than tailoring a bid to each auction, a bidder scales \emph{all} of her values by a single multiplicative factor, her \emph{pacing multiplier}, and submits the resulting shaded values as bids everywhere. Her entire strategy then collapses to the choice of a single number. In practice, advertisers typically also face budget or return-on-spend constraints, and autobidders adjust their pacing multipliers over time to meet them, often using no-regret learning algorithms~\citep{aggarwal2019autobidding,balseiro2019learning}.

This simplification does not remove strategic reasoning so much as concentrate it into that one choice. Because a single multiplier governs every bid at once, the auctions become coupled: raising the multiplier to capture a contested impression also raises the price paid on all items the bidder is already winning, while lowering it to save on spending might result in losing some item. The multiplier that is best for a bidder therefore depends on those chosen by her competitors, and the stable outcomes of this interaction are precisely the equilibria of a game played through pacing multipliers. It is this game, in the first-price setting, that we study in this paper.

The questions we ask about it are the classical ones of algorithmic game theory: do pure Nash equilibria always exist, can we decide if one exists/compute it efficiently, and how much social welfare could be lost at equilibrium? For pacing in advertising markets similar questions have received a great deal of attention, but almost always through a \emph{market} lens rather than a strategic one. The most influential such formulation is the \emph{First-Price Pacing Equilibrium} (FPPE) of \citet{conitzer2022pacing}, in which bidders are budget-constrained and each pacing multiplier is not a free strategic choice but is fixed by a market-clearing condition that ties a bidder's spending to her budget. 

We take the complementary, game-theoretic view: we treat the pacing multiplier as the genuine strategic choice of a \emph{utility-maximizing} bidder, who balances the value of the items she wins against the price she pays for them. This gives rise to a very rich strategic game, even in the absence of budget constraints. Asking our three questions for the resulting full-information game, we find that the answers differ sharply from those of the market setting: equilibria may fail to exist, deciding whether they do is computationally hard, and, even when they exist, they can be inefficient. Despite the rich literature on pacing and on the efficiency of auctions, surveyed in \cref{sec:related}, to the best of our knowledge this basic strategic model has not been studied before.

\subsection{Our Contribution}

We introduce and study \emph{Auctions with Pacing Strategies (APS) games}, the full-information game in which $n$ utility-maximizing bidders each choose a single pacing multiplier to shade their values across $m$ simultaneous first-price auctions. We consider three fundamental questions often posed in algorithmic game theory, related to \emph{existence}, \emph{computational complexity}, and \emph{inefficiency} of (pure) Nash equilibria. Our results completely settle all three questions. In more detail:

\begin{itemize}

\item[-] First, on \emph{existence}, we show that a pure Nash equilibrium need not exist, already with two bidders and two items. This non-existence is robust to approximation: no $\varepsilon$-approximate PNE is guaranteed for any $\varepsilon$ up to roughly $m/25$, where $m$ is the number of items of the APS instance.

\item[-] Second, on \emph{computational complexity}, we show that deciding whether an APS game admits an (approximate) PNE is NP-complete. The hardness is again robust to approximation, and in fact survives a \emph{super-constant} slack: it persists for $\varepsilon = \Theta(m^{1-\delta})$, for every constant $\delta \in (0,1)$. On the positive side, we give polynomial-time algorithms whenever either the number of bidders $n$ or the number of items $m$ is held constant, further clarifying the tractability boundary of the problem.

\item[-] Third, on \emph{efficiency}, we exactly characterize the social welfare loss at equilibrium whenever one exists. We prove that both the Price of Anarchy and the Price of Stability equal $\frac{e}{e-1}$, so every equilibrium recovers at least a $(1-1/e)$ fraction of the optimal welfare, and no choice of equilibrium can do better in the worst case.

\end{itemize}

Taken together, these results show that pacing alone already makes the strategic problem rich: deciding whether an equilibrium exists is NP-complete in general, and even when the number of bidders or items is fixed, it requires involved polynomial-time algorithms. Adding budgets cannot make the problem easier, since our hardness construction embeds directly into the budgeted setting as an instance in which no budget constraint binds; we discuss budgets further in the concluding section.

\subsection{Related Work}\label{sec:related}

Our work lies at the intersection of three strands of research: pacing in online advertising markets, the existence and computation of equilibria in such markets, and the inefficiency of (simultaneous) auctions.
We survey each in turn and position our contributions accordingly.

\paragraph{Budget management and pacing.}
The use of a single multiplicative \emph{pacing multiplier} to throttle a bidder's bids is the dominant budget-management heuristic deployed by autobidding agents in practice, and it is the strategy space we study. The idea goes back to the dynamics of bid optimization studied by \citet{borgs2007dynamics}, who showed that simple multiplicative bid-adjustment rules used to smooth budget expenditure are tightly connected to market equilibria. A substantial line of work analyzes repeated auctions with budgets from a learning and mean-field perspective: \citet{balseiro2015repeated} characterize equilibria of budget-constrained ad exchanges via a fluid (mean-field) relaxation, and \citet{balseiro2019learning} show that adaptive pacing strategies are regret-minimizing and converge to a market equilibrium. \citet{conitzer2022multiplicative} formalized \emph{multiplicative pacing equilibria} for second-price markets, in which each bidder uniformly shades her values by a multiplier (precisely the strategic primitive we adopt), and \citet{conitzer2022pacing} developed the analogous theory of \emph{first-price pacing equilibria} (FPPE).
A key conceptual difference is that much of this literature models bidders either as following a fixed pacing heuristic or as \emph{value-maximizers} subject to budget/return constraints, with the multiplier emerging from a budget-feasibility (market-clearing) condition. As a result, a pacing equilibrium is in essence a market equilibrium, rather than a Nash equilibrium of a fully strategic game, see \citet{conitzer2022pacing} for an interesting discussion on the connection between the two. In contrast, we treat the pacing multiplier as the strategic choice of a \emph{utility-maximizing} (quasi-linear) bidder and study the pure Nash equilibria of the resulting full-information game.
\paragraph{Efficiency in the autobidding world.}
A recent and closely related body of work studies the efficiency of auctions when bidders are \emph{value-maximizing} autobidders subject to return-on-spend or budget constraints. Because of these constraints, the standard truthfulness of 2nd-price and VCG mechanisms breaks down, forcing strategic deviations that create inefficiency. This sharply contrasts with our utility-maximizing setting. 
\citet{aggarwal2019autobidding} introduced this model of autobidding with constraints and established a Price of Anarchy of $2$ for truthful (second-price) auctions, and \citet{deng2021towards} extended the factor-$2$ guarantee to VCG position auctions while showing that it can be improved below $2$ using value-based boosts and reserves. For the non-truthful \emph{first-price} format, which is closest to ours, \citet{liaw2023efficiency} proved that the Price of Anarchy under return-on-spend constraints is exactly $2$, and \citet{deng2022efficiency} generalized this to randomized bidding strategies and to mixed populations of value- and utility-maximizing agents, where the Price of Anarchy rises to roughly $2.188$; \citet{mehta2022auction} further showed that randomization on the auctioneer's side can push efficiency below the deterministic factor of $2$. Finally, for a comprehensive overview of the recent shift towards autobidding and its implications for auction theory, we refer the reader to the survey by \citet{aggarwal2024auto}.
Closer to the efficiency questions we study, \citet{colinibaldeschi2026liquid} recently introduced a type-dependent smoothness framework for autobidders with budget and return-on-spend constraints, and used it to obtain tight liquid welfare guarantees for simultaneous first-price auctions that depend on the mix of agent types present.

\paragraph{Existence and computation of equilibria.}
The existence and tractability of pacing equilibria depend sharply on the underlying auction format. \citet{conitzer2022pacing} show that first-price pacing equilibria always exist, are essentially unique, and can be computed in polynomial time via a convex (Eisenberg--Gale-type) program. The second-price case is markedly different: 
\citet{conitzer2022multiplicative} show that second-price pacing equilibria always exist but  computing either a social-welfare-maximizing or a revenue-maximizing pacing
equilibrium is NP-hard. Later \citet{chen2023pacing} proved that computing even an approximate pacing equilibrium for second-price auctions is PPAD-complete, and \citet{chenli2025constant} strengthened this to constant-factor approximations. 
Independently of our results, \citet{yan2026fewbuyers} and \citet{huang2026fewgoods} developed a cell decomposition of the space of pacing multipliers close in spirit to ours, for a different setting: the second-price pacing equilibria of \citet{conitzer2022multiplicative}, in which the pacing multipliers of budget-constrained bidders in second-price auctions are determined by a market-clearing condition. They use it to compute such an equilibrium exactly in polynomial time when, respectively, the number of bidders or the number of items is constant.
The methods are therefore related, but they are applied to a different problem: a second-price pacing equilibrium is known to always exist, so the task there is to compute one, whereas in our strategic setting an equilibrium may be absent and the question we answer is whether one exists at all. The computational complexity of equilibrium computation in other first-price mechanisms has also been studied extensively, albeit predominantly in the Bayesian (incomplete-information) setting for the sale of a single item, where the analysis differs substantially from ours, e.g., see \citep{chen2023complexity,filosratsikas2025correlated,filosratsikas2026symmetric,FGHK24_sicomp,fghlp2021_sicomp} and references therein.  

\paragraph{Price of Anarchy and smoothness.}
The concept of the Price of Anarchy was introduced by \citet{koutsoupias1999worst} as a measure of the worst-case efficiency loss incurred by selfish behaviour at equilibrium and has since become a central tool in algorithmic game theory. The efficiency of equilibria in auctions is often studied via the \emph{smoothness} framework, where \citet{roughgarden2015intrinsic} established that Price of Anarchy bounds proved via smoothness extend automatically to mixed Nash, correlated and coarse correlated equilibria, and \citet{syrgkanis2013composable} showed that simultaneous first-price item auctions are $(1-1/e)$-efficient, i.e., have a Price of Anarchy of $e/(e-1)$. The Price of Anarchy of simultaneous item auctions, including in the Bayesian setting, has been studied extensively, e.g., see~\citep{bhawalkar2011welfare,christodoulou2016bayesian,christodoulou2016tight,feldman2013simultaneous,syrgkanis2013composable} and the survey by \citet{roughgarden2017price}; for the Bayesian first-price auction of a single item, tight bounds on both the Price of Anarchy and the Price of Stability are now known~\citep{hoy2018tighter,jin2023first,jin2023stability,syrgkanis2013composable}. Our work can be seen as a variant of the simultaneous item bidding setting studied by \citet{feldman2013simultaneous}, \citet{christodoulou2016tight} and \citet{voudouris2020simple}, specialized to additive valuations, 
where bidders are restricted to pacing multipliers rather than choosing independent bids in every auction.

\section{Preliminaries}
An instance $\mathcal I$ of an \emph{Auctions with Pacing Strategies (APS)} problem consists of a set $\mathcal{N}$ of $n$ bidders, participating in $m$ parallel auctions $A_1,A_2,\ldots,A_m$, where 
each auction follows a first-price auction format.
Each bidder $i$ has a value $v_{i,j} \in [0, 1]$ for each item $j$, which is public information.
The strategy of bidder $i$ is then encoded by a single parameter $\alpha_i \in [0,1]$, also called a \emph{pacing multiplier}, which scales down her value for all items, such that her bid for item $j$ is $b_{i,j}=\alpha_i v_{i,j}$. Let $\vec \alpha = (\alpha_1, \ldots, \alpha_n)$ be the strategy profile and $\vec b$ denote the matrix of all the bids. 

\paragraph{First-Price Auctions.}
The mechanism runs $m$ parallel first-price auctions meaning that each item is assigned to a single highest bidder, who pays her own bid. Formally, let
\[
    W_j \coloneq \argmax_{i \in \mathcal{N}} b_{i,j}
\]
be the set of highest bidders on item $j$. A \emph{tie-breaking rule} selects a single winner $w_j \in W_j$.
We will fix a concrete such rule below.
The allocation and price of item $j$ are then
\begin{equation}
    \label{eq:allocation}
    x_{i,j} \coloneq
    \begin{cases}
        1, & \text{if } i = w_j, \\
        0, & \text{otherwise},
    \end{cases}
    \qquad
    p_j \coloneq \max_{i \in \mathcal{N}} b_{i,j},
\end{equation}
so that the winner $w_j$ pays $p_j = b_{w_j,j}$ and every other bidder pays $0$.

\paragraph{Tie-breaking.}
Since the set of pacing multipliers comes from a continuous space, the space of possible bids is also continuous.
While ties are unlikely to happen, since we are studying the complete information game, it is still necessary to pick a tie-breaking rule to be used in the computation of utilities.
We will assume that the auction employs the \emph{lexicographic} tie-breaking rule, restricted to the bidders that value the item positively: given the fixed ordering of the bidders $1, 2, \ldots, n$, whenever several bidders are tied for an item, it is awarded to the one with the smallest index among those with a positive value for it. Formally, this instantiates the winner of~\eqref{eq:allocation} as
\begin{equation}
    \label{eq:lex-winner}
    w_j \coloneq \min \sset{i \in W_j \fwhs v_{i,j} > 0},
\end{equation}
who is then the only bidder paying for item $j$; an item that no bidder values positively is left unallocated. Note that whenever the highest bid on an item is positive, every tied bidder values it positively (since $b_{i,j} = \alpha_i v_{i,j} > 0$ implies $v_{i,j} > 0$), so the restriction only matters for items on which all bids are $0$: such an item goes to the lowest-indexed bidder that values it, rather than to a bidder with no value for it. Equivalently, one may think of each bidder as participating only in the auctions for the items she values.

\paragraph{Utilities and Equilibrium.}
This setting can be viewed as a game of complete information, where each bidder $i$ chooses her pacing multiplier $\alpha_i$ aiming to maximize her utility from participating in the auctions.
Here the utility of a bidder $i$ is defined as:
\begin{equation}
    u_i(\alpha_i, \vec{\alpha_{-i}}) \coloneq \sum_{j=1}^m x_{i,j} \cdot (v_{i,j} - b_{i,j}) = \sum_{j=1}^m x_{i,j}\cdot v_{i,j} \cdot (1- \alpha_{i}),
\end{equation} 
 where $\vec{\alpha_{-i}}$ is the strategy profile of all the bidders except for bidder $i$. We call this the \emph{Auctions with Pacing Strategies (APS) game}.

In such games, given $\vec{\alpha_{-i}}$, a \emph{best-response} of bidder $i$ is a multiplier that maximizes her utility. Since the space of pacing multipliers is continuous while the allocation jumps at ties, such a maximizer need not exist: to take an item away from a bidder who wins ties against her, bidder $i$ has to bid strictly higher, and there is no smallest such bid. Her \emph{best-response payoff}
\begin{equation}\label{eq:br-payoff}
    u_i^{\ast}(\vec{\alpha_{-i}}) \coloneq \sup_{\alpha_i' \in [0,1]} u_i(\alpha_i', \vec{\alpha_{-i}})
\end{equation}
is nevertheless always well defined, and \cref{lem:deviation-sup} below shows that it is given by a simple formula. We are interested in the notion of an (approximate) equilibrium of such games, which is a strategy profile in which every bidder is best-responding up to an $\varepsilon$. It is formally defined as follows.

\begin{definition}[($\varepsilon$-approximate) Pure Nash Equilibrium of the APS game]
   For $\varepsilon \geq 0$,
   a strategy profile $\vec{\alpha}$ consisting of the bidders' pacing multipliers is an $\varepsilon$-approximate Pure Nash Equilibrium ($\varepsilon$-PNE) if, for any player $i$ and for any value of $\alpha_i' \in [0,1]$, it holds that:
   \begin{equation}
       u_i(\vec{\alpha}) \geq u_i(\alpha_i',\vec{\alpha_{-i}}) - \varepsilon.
   \end{equation}
\end{definition}
\noindent We will refer to a $0$-PNE as an \emph{exact} PNE of the game, and use $\mathsf{NE}(\mathcal I)$ to indicate the set of all PNE in an instance $\mathcal I$.

\paragraph{Computing best-responses.}
We now record how a bidder's best-response payoff is computed, showing that it does not depend on how the auction breaks ties.

Fix a bidder $i$ and a profile $\vec{\alpha_{-i}}$ of the remaining bidders, and let
\begin{equation}\label{eq:threshold}
    t_{i,j} \coloneq \frac{\max_{k \neq i} \alpha_k v_{k,j}}{v_{i,j}}
\end{equation}
be the smallest multiplier with which $i$ matches the highest competing bid on item $j$, with the convention $t_{i,j} = \infty$ when $v_{i,j} = 0$, so that an item she does not value is never acquired.
Let $\pi_i$ be a permutation of the items ordering these thresholds, so that $t_{i,\pi_i(1)} \leq \dots \leq t_{i,\pi_i(m)}$, let $S_{i,k} = \{\pi_i(1), \dots, \pi_i(k)\}$ be the prefix of the first $k$ of them, and set $t_{i,\pi_i(0)} \coloneq 0$, so that $S_{i,0} = \varnothing$.
Finally, write $V_i(S) = \sum_{j \in S} v_{i,j}$ for bidder $i$'s total value for a set $S$ of items.

\begin{lemma}\label{lem:deviation-sup}
    For every bidder $i$ and every profile $\vec{\alpha_{-i}}$, the best-response payoff of bidder $i$ is
    \begin{equation}\label{eq:deviation-sup}
        u_i^{\ast}(\vec{\alpha_{-i}})
        \;=\; \max_{k \,:\, t_{i,\pi_i(k)} \leq 1} \left(1 - t_{i,\pi_i(k)}\right) V_i(S_{i,k}).
    \end{equation}
    The right-hand side is the utility that bidder $i$ obtains when the ties she is involved in are resolved in her favour, and it does not depend on the tie-breaking rule employed by the auction.
\end{lemma}

\begin{proof}
    As $\alpha_i'$ increases from $0$ to $1$, bidder $i$ acquires the items one by one in increasing order of their thresholds: she wins item $j$ whenever $\alpha_i' > t_{i,j}$ and loses it whenever $\alpha_i' < t_{i,j}$.
    The set she wins at $\alpha_i'$ is therefore contained in the set of items whose threshold is at most $\alpha_i'$, which, since $\pi_i$ orders the thresholds, is a prefix $S_{i,k}$ with $t_{i,\pi_i(k)} \leq \alpha_i'$; how the ties at $\alpha_i'$ itself fall is the only freedom, and it can only remove items from that prefix.
    Her utility is then at most $(1-\alpha_i')V_i(S_{i,k}) \leq (1 - t_{i,\pi_i(k)})V_i(S_{i,k})$, which is at most the right-hand side of~\eqref{eq:deviation-sup}; since $\alpha_i' \leq 1$, only prefixes with $t_{i,\pi_i(k)} \leq 1$ arise.

    For the converse, fix such a prefix $S_{i,k}$. If the ties at $t_{i,\pi_i(k)}$ are resolved in favour of $i$, then playing $\alpha_i' = t_{i,\pi_i(k)}$ wins her at least $S_{i,k}$ and yields exactly $(1 - t_{i,\pi_i(k)})V_i(S_{i,k})$. If they are resolved against her and $t_{i,\pi_i(k)} < 1$, then every $\alpha_i'$ slightly above $t_{i,\pi_i(k)}$ wins her a set containing $S_{i,k}$, so her utility is at least $(1-\alpha_i')V_i(S_{i,k})$, which tends to $(1 - t_{i,\pi_i(k)})V_i(S_{i,k})$ as $\alpha_i'$ decreases to $t_{i,\pi_i(k)}$. If instead $t_{i,\pi_i(k)} = 1$, the prefix contributes the value $0$, which is a lower bound on the supremum in any case.
    In both cases the supremum equals the right-hand side of~\eqref{eq:deviation-sup}, independently of how the ties were resolved.
\end{proof}

\noindent In particular, checking that bidder $i$ is best-responding up to $\varepsilon$ amounts to comparing $u_i(\vec{\alpha})$ with at most $m+1$ prefix values.
Whenever we compute a best-response in what follows, we therefore award the bidder any item she ties for; by \cref{lem:deviation-sup} this is not an assumption on the auction, but merely the way her best-response payoff is written.
A prefix with $t_{i,\pi_i(k)} > 1$ is not a feasible deviation and is excluded from~\eqref{eq:deviation-sup}.
When we later write the equilibrium condition as one linear constraint per prefix, such prefixes may nonetheless be kept without harm: $t_{i,\pi_i(k)} > 1$ makes the right-hand side of the corresponding constraint non-positive, while $u_i(\vec{\alpha}) \geq 0$ always, so the constraint is satisfied automatically.

We will also study the efficiency of such systems due to the selfish behaviour of the bidders, with respect to social welfare and through the lens of Price of Anarchy and Price of Stability. We postpone the formal definition of related notions to \Cref{sec:inefficiency}.

\section{Complexity of Equilibrium Computation}\label{sec:complexity}

In this section, first we show that APS games need not admit an (approximate) pure Nash equilibrium, and that the largest $\varepsilon$ for which an $\varepsilon$-PNE does not necessarily exist grows \emph{linearly} with the number of items. After that we discuss the hardness of deciding whether a given instance has an $\varepsilon$-PNE in general and some special cases.

We first state a simple transformation of APS instances, replication together with padding, which we use to amplify both our non-existence result and our hardness result.
Given an APS instance $\mathcal{I}$ with bidder set $\mathcal{N}$ and item set $M$, and integers $r \geq 1$ and $p \geq 0$, let $\mathcal{I}^{(r,p)}$ be the instance with the same bidder set, item set $(M \times [r]) \cup P$ where $\cards{P} = p$, and values
\[
    v_{i,(j,c)} = v_{i,j} \quad \text{for all } i \in \mathcal{N},\, j \in M,\, c \in [r],
    \qquad
    v_{i,q} = 0 \quad \text{for all } i \in \mathcal{N},\, q \in P;
\]
that is, $\mathcal{I}^{(r,p)}$ consists of $r$ item-disjoint copies of $\mathcal{I}$, padded with $p$ items that no bidder values.

\begin{lemma}[Replication]\label{lem:replication}
    Let $\mathcal{I}$ be an APS instance, let $r \geq 1$ and $p \geq 0$ be integers, and let $\vec{\alpha} \in [0,1]^n$ be any strategy profile. Then every bidder's utility satisfies $u_i^{\mathcal{I}^{(r,p)}}(\vec{\alpha}) = r \cdot u_i^{\mathcal{I}}(\vec{\alpha})$, and consequently, for every $\varepsilon \geq 0$, the set of $\varepsilon$-PNE of $\mathcal{I}^{(r,p)}$ equals the set of $\frac{\varepsilon}{r}$-PNE of $\mathcal{I}$.
    In particular, if $\mathcal{I}$ admits no $\delta$-PNE then $\mathcal{I}^{(r,p)}$ admits no $\varepsilon$-PNE for any $\varepsilon \leq r\delta$, and $\vec{\alpha}$ is an exact PNE of $\mathcal{I}$ if and only if it is an exact PNE of $\mathcal{I}^{(r,p)}$.
\end{lemma}

Replication does not enlarge the strategy space: a bidder still plays a single multiplier, which acts on all $r$ copies simultaneously, so she cannot behave differently in different copies and every deviation gain is multiplied by exactly $r$. The two sets of equilibria are equal, and not merely related by inclusion, so replication transfers both a positive and a negative answer. The padding items are valued at $0$ by every bidder, so they are never allocated and only serve to fix the number of items.

\begin{proof}
    Fix a profile $\vec{\alpha}$, and consider a copy $c \in [r]$ of an item $j \in M$.
    The bid of bidder $i$ on the item $(j,c)$ is $\alpha_i v_{i,(j,c)} = \alpha_i v_{i,j}$, exactly her bid on $j$ in $\mathcal{I}$; in particular, the set of highest bidders on $(j,c)$ and the set of bidders valuing it positively coincide with the corresponding sets for $j$ in $\mathcal{I}$.
    Since the lexicographic rule~\eqref{eq:lex-winner} selects the winner using only the bidders' indices and these two sets, it awards $(j,c)$ to the winner of $j$, at the same price.
    Every copy therefore reproduces the allocation and the prices of $\mathcal{I}$ under $\vec{\alpha}$, while a padding item $q \in P$, being valued at $0$ by everyone, is left unallocated and contributes neither value nor payment to any bidder.
    As utilities are additively separable across items, summing over the $r$ copies and the $p$ padding items gives $u_i^{\mathcal{I}^{(r,p)}}(\vec{\alpha}) = r \cdot u_i^{\mathcal{I}}(\vec{\alpha})$.

    Crucially, the strategy space is unchanged: in $\mathcal{I}^{(r,p)}$ bidder $i$ still commits to a single multiplier and thus cannot treat copies differently, so the identity applies verbatim to any unilateral deviation $\alpha_i'$. Every deviation gain is therefore scaled by exactly the same factor:
    \[
        u_i^{\mathcal{I}^{(r,p)}}(\alpha_i',\vec{\alpha_{-i}}) - u_i^{\mathcal{I}^{(r,p)}}(\vec{\alpha})
        = r\left(u_i^{\mathcal{I}}(\alpha_i',\vec{\alpha_{-i}}) - u_i^{\mathcal{I}}(\vec{\alpha})\right).
    \]
    Hence no deviation gains more than $\varepsilon$ in $\mathcal{I}^{(r,p)}$ if and only if no deviation gains more than $\frac{\varepsilon}{r}$ in $\mathcal{I}$, which is precisely the claimed equivalence.
    The two special cases follow by taking $\varepsilon \leq r\delta$ and $\varepsilon = 0$, respectively.
\end{proof}

\subsection{Non-existence of equilibrium}

The core of our proof to show the non-existence of equilibrium is a single two-item gadget that on its own already precludes an $\varepsilon$-PNE for $\varepsilon\le\frac{2}{25}$ (the case $k=1$ below); we then apply \cref{lem:replication} to this gadget to amplify the threshold.

For impossibility results of this kind, the aim usually is to rule out an $\varepsilon$-PNE for a \emph{constant} $\varepsilon$. Here, however, values are normalized \emph{per item} ($v_{i,j}\in[0,1]$), which places a bidder's utility on a scale of $\Theta(m)$ rather than $\Theta(1)$; the additive slack $\varepsilon$ should be measured against this scale, so the natural analog of a constant-$\varepsilon$ impossibility is one in which $\varepsilon$ grows linearly with $m$. Equivalently, rescaling each instance's values so that the utilities lie in $[0,1]$ turns our bound into a constant. Note that this linear growth is asymptotically the maximum we can hope for: any bidder's utility, and hence any unilateral deviation gain, is at most her total value $\sum_j v_{i,j}\le m$, so every profile is trivially an $\varepsilon$-PNE once $\varepsilon\ge m$; non-existence can therefore hold only for $\varepsilon=O(m)$.

\begin{theorem}\label{thm:linear-nonexistence}
    For any positive integer $k$ and any $\varepsilon\leq\frac{2k}{25}$, APS games do not always admit an $\varepsilon$-PNE, even for the case of $2$ bidders and $2k$ items.
    More generally, for every number of items $m\geq 2$ there are $2$-bidder instances with $m$ items admitting no $\varepsilon$-PNE for any $\varepsilon\leq\frac{2}{25}\left\lfloor\frac{m}{2}\right\rfloor$, that is, $\frac{m}{25}$ for even $m$ and $\frac{m-1}{25}$ for odd $m$; the non-existence threshold therefore grows linearly in the number of items.
\end{theorem}
\begin{proof}
    Consider \emph{the base instance}
   to be an APS game with $2$ bidders and $2$ items, where the bidders' values for the items are $\vec{v_1}=(1/2,1)$, $\vec{v_2}=(1/2,3/10)$.
    We first show that this instance admits no $\varepsilon$-PNE for $\varepsilon \leq \frac{2}{25}$.
    
    To reason about the bidders being at equilibrium, we will be comparing their strategies to their best-response.
    Since the action space is continuous, it is obvious that any best-response of a bidder will result in her bidding marginally higher than at least $1$ other bidder for some item, to win that item at the lowest possible price.
    To facilitate our analysis, we will be computing the value of the multiplier that results in a \emph{tie} for that item, awarding the item to the bidder that is computing her best-response.
    By \cref{lem:deviation-sup}, this is exactly her best-response payoff, whether or not the tie-breaking rule awards her the item.
        
    To begin with, notice that bidder $1$'s best-response is always to either tie bidder $2$ for the first item ($\alpha_1=\alpha_2$) or for the second item $\alpha_1=\frac{3}{10}\alpha_2$.
    Similarly, for bidder $2$'s best-response, it must be the case that either $\alpha_2=\alpha_1$ or $\alpha_2=\frac{10}{3}\alpha_1$.
    
    We start with the following observation.

    \begin{observation}
        At any $\frac{2}{25}$-PNE of this instance, each bidder wins exactly $1$ item.
    \end{observation}

    \begin{proof}
        Assume for contradiction that one of the bidders wins both items at an equilibrium $(\alpha_1,\alpha_2)$.
        We consider two different cases in which this can happen:

        \begin{enumerate}
            \item Bidder $2$ wins both items, meaning that $\alpha_2\geq \frac{10}{3}\alpha_1$.

            Consider the equilibrium condition from the point of view of bidder $1$, comparing to her deviation to a multiplier of $\tfrac{3}{10}\alpha_2$:
            \begin{equation*}
                u_1(\alpha_1,\alpha_2) \geq u_1\left(\frac{3}{10}\alpha_2, \alpha_2\right)-\frac{2}{25} \Rightarrow 0 \geq \left( 1- \frac{3}{10}\alpha_2\right) - \frac{2}{25} \Rightarrow \alpha_2 \geq \frac{46}{15}
            \end{equation*}
            which contradicts the fact that $\alpha_2 \leq 1$.
            
            \item Bidder $1$ wins both items, meaning that $\alpha_1 \geq \alpha_2$.

            The utility of bidder $1$ in this case is $u_1(\alpha_1,\alpha_2)=(1-\alpha_1)\frac{3}{2}$.

            Using the equilibrium condition:

            \begin{equation}\label{eq:bidder-1}
                u_2(\alpha_1,\alpha_2) \geq u_2(\alpha_1,\alpha_1) - \frac{2}{25} \Rightarrow (1-\alpha_1)\frac{1}{2} \leq \frac{2}{25} \Rightarrow \alpha_1 \geq \frac{21}{25}
            \end{equation}

            Using the equilibrium condition again, from the point of view of player $1$ this time.
            \[
            u_1(\alpha_1,\alpha_2) \geq u_1\left(\frac{3}{10}\alpha_2,\alpha_2\right)-\frac{2}{25} \Rightarrow (1-\alpha_1)\frac{3}{2} \geq \left(1-\frac{3}{10}\alpha_2\right)-\frac{2}{25} \Rightarrow \alpha_1 \leq \frac{29}{60}
            \]
            which contradicts~\eqref{eq:bidder-1}.
        \end{enumerate}
    \end{proof}

    This directly yields the following:

    \begin{corollary}
       At any $\frac{2}{25}$-PNE, bidder $2$ wins the first item and bidder $1$ wins the second one, or, equivalently, $\frac{10}{3}\alpha_1 \geq \alpha_2 > \alpha_1$.
    \end{corollary}

    Now, assume for contradiction that such an equilibrium $(\alpha_1,\alpha_2)$ exists.
    Then, by the equilibrium condition, comparing bidder $1$'s utility to her deviation to $\alpha_2$, which would win her both items, the following must hold:

    \begin{equation}\label{eq:bidder-1-1item}
        u_1(\alpha_1,\alpha_2) \geq u_1(\alpha_2,\alpha_2)-\frac{2}{25} \Rightarrow 1-\alpha_1 \geq (1-\alpha_2)\frac{3}{2} -\frac{2}{25} \Rightarrow \alpha_2 \geq \frac{2}{3}\alpha_1 + \frac{7}{25}
    \end{equation}

    Moreover, bidder $1$ must not profit from lowering her multiplier to $\frac{3}{10}\alpha_2$, which would still win her the second item (by tie-breaking), but at a lower price:
    \begin{equation}\label{eq:bidder-1-undercut}
        u_1(\alpha_1,\alpha_2) \geq u_1\left(\frac{3}{10}\alpha_2,\alpha_2\right)-\frac{2}{25} \Rightarrow 1-\alpha_1 \geq \left(1-\frac{3}{10}\alpha_2\right) -\frac{2}{25} \Rightarrow \alpha_1 \leq \frac{3}{10}\alpha_2 + \frac{2}{25}
    \end{equation}

    Looking at the equilibrium condition from the point of view of the second player:
    \begin{equation}\label{eq:bidder-2-1item}
        u_2(\alpha_1,\alpha_2) \geq u_2\left(\alpha_1,\alpha_1\right)-\frac{2}{25} \Rightarrow (1-\alpha_2)\frac{1}{2} \geq \left( 1- \alpha_1 \right) \frac{1}{2}-\frac{2}{25} \Rightarrow \alpha_2 \leq \alpha_1 + \frac{4}{25}
    \end{equation}

    Combining~\eqref{eq:bidder-1-1item} and~\eqref{eq:bidder-2-1item}, we get:
    \[
    \frac{2}{3}\alpha_1 + \frac{7}{25} \leq \alpha_2 \leq \alpha_1 + \frac{4}{25} \Rightarrow \alpha_1 \geq \frac{9}{25},
    \]
    while combining~\eqref{eq:bidder-1-undercut} and~\eqref{eq:bidder-2-1item} yields:
    \[
    \alpha_1 \leq \frac{3}{10}\left(\alpha_1 + \frac{4}{25}\right) + \frac{2}{25} \Rightarrow \alpha_1 \leq \frac{32}{175}.
    \]
    Since $\frac{32}{175} < \frac{9}{25}$, the two bounds are incompatible, leading to a contradiction.
    This establishes that the base instance admits no $\frac{2}{25}$-PNE.

    \medskip
    \emph{Amplification and padding.}
    Fix any number of items $m \geq 2$, and set $r = \lfloor m/2 \rfloor \geq 1$ and $p = m - 2r \in \{0,1\}$.
    Applying \cref{lem:replication} to the base instance $\mathcal{I}$ gives the instance $\mathcal{I}^{(r,p)}$, which has $2$ bidders and exactly $2r + p = m$ items.
    Since $\mathcal{I}$ admits no $\frac{2}{25}$-PNE, the lemma yields that $\mathcal{I}^{(r,p)}$ admits no $\varepsilon$-PNE for any $\varepsilon \leq \frac{2r}{25} = \frac{2}{25}\left\lfloor\frac{m}{2}\right\rfloor$, which equals $\frac{m}{25}$ for even $m$ and $\frac{m-1}{25}$ for odd $m$.
    Specialising to $p = 0$ and $m = 2r$ recovers the first statement of the theorem, with $k = r$ copies.
\end{proof}

\subsection{NP-completeness}

Next, we study the computational complexity of the problem of deciding whether an (approximate) equilibrium exists in a given instance.
We establish that the problem is NP-complete in general, and that this remains true even when the additive slack $\varepsilon$ is allowed to grow with the number of items.
We begin by proving that the problem belongs in NP.

\begin{theorem}\label{thm:np-membership}
    Deciding whether an APS game admits an $\varepsilon$-PNE is in NP, for any rational $\varepsilon \ge 0$ given as part of the input.
\end{theorem}

\begin{proof}
    The certificate is a strategy profile $\vec{\alpha} \in \big([0,1] \cap \mathbb{Q}\big)^n$. We show that such a profile can be verified in polynomial time, and that whenever an $\varepsilon$-PNE exists, one of polynomial bit-length does too.

    \emph{Verification.} Given $\vec{\alpha}$, each item is awarded to the highest paced bid $\alpha_i v_{i,j}$ (with lexicographic tie-breaking), so the allocation and the utilities $u_i(\vec{\alpha})$ are computed in $\mathcal{O}(nm)$ time. To check the equilibrium condition for a bidder $i$, we compute the thresholds $t_{i,j}$ of~\eqref{eq:threshold} and sort them, which by \cref{lem:deviation-sup} gives her best-response payoff as the largest of the $\mathcal{O}(m)$ prefix values $(1-t_{i,\pi_i(k)})V_i(S_{i,k})$; evaluating them yields $u_i^{\ast}(\vec{\alpha}_{-i})$ in polynomial time. We accept iff
    \[
        u_i(\vec{\alpha}) \ge u_i^{\ast}(\vec{\alpha}_{-i}) - \varepsilon \qquad \text{for every bidder } i,
    \]
    which is exactly the $\varepsilon$-PNE condition. The entire check runs in time polynomial in $n$, $m$, and the input bit-length.

    \emph{A short certificate exists.} Suppose some $\varepsilon$-PNE $\vec{\alpha}^\ast$ exists. It induces a discrete structure: for every item $j$, the order of the bids $\alpha^\ast_i v_{i,j}$ of the bidders that value $j$, in which some bids may be equal, and for every bidder $i$, the order of her thresholds $t_{i,j}$ on the items she values. The order of the bids determines the winner of every item under the lexicographic rule, as well as the highest bid that every bidder faces on it, so each threshold is the bid of a fixed bidder divided by $v_{i,j}$, a linear function of $\vec{\alpha}$. With this structure held fixed, the conditions defining an $\varepsilon$-PNE become \emph{linear} in $\vec{\alpha}$. They are of three kinds: the equalities and inequalities that fix the order of the bids on each item; the inequalities encoding the fixed threshold orderings; and the best-response inequalities, one set per bidder, which bound $u_i(\vec{\alpha})$ from below by the utility of each of her prefix deviations (now linearly expressible), up to $\varepsilon$. Some of the inequalities of the first kind are strict, and we handle them with the slack-maximization step used in the proof of \cref{thm:const-m-poly}: we replace every strict inequality $a > b$ by $a \geq b + \delta$ for a single new variable $\delta \leq 1$, and maximize $\delta$. Every coefficient is a product or ratio of the rational valuations and of $\varepsilon$, and is therefore of polynomial bit-length. The resulting linear program is feasible with a positive objective value, as $\vec{\alpha}^\ast$ together with a small enough $\delta > 0$ satisfies it, and its optimum is attained at a vertex whose bit-length is polynomial in that of its coefficients~\citep{schrijver1986theory}. At this vertex $\delta > 0$, so the order of the bids on every item is the same as at $\vec{\alpha}^\ast$, and so are the allocation and the thresholds as linear functions of $\vec{\alpha}$. The vertex is therefore again an $\varepsilon$-PNE, and serves as a polynomial-size certificate.
\end{proof}

We now turn to hardness. Our main result is that hardness is not confined to a constant additive slack: it persists even when $\varepsilon$ is allowed to grow almost linearly with the number of items.

\begin{theorem}\label{thm:eps-np-hard}
    For every constant $\delta \in (0,1)$, deciding whether an APS game admits an $\varepsilon$-PNE is NP-hard, already on instances with $m$ items and $\varepsilon = \Theta(m^{1-\delta})$.
\end{theorem}

Combining this with the membership result of \cref{thm:np-membership}, we obtain the claimed NP-completeness.

\begin{corollary}\label{cor:np-complete}
    Deciding whether an APS game admits an $\varepsilon$-PNE, for a rational $\varepsilon \geq 0$ given as part of the input, is NP-complete. This holds already for $\varepsilon = \Theta(m^{1-\delta})$, for every constant $\delta \in (0,1)$.
\end{corollary}

The proof has two steps. We first reduce from \emph{Exact Cover by 3-Sets} (X3C) to obtain a \emph{hard core}, an instance on which the two answers are separated by a \emph{constant} gap. We then apply \cref{lem:replication} to that core to inflate the gap. Recall that in X3C we are given a universe $U = \{e_1, \dots, e_{3k}\}$ and a collection $\mathcal{S} = \{S_1, \dots, S_M\}$ of $3$-element subsets of $U$, and we must decide whether some subcollection $\mathcal{S}^\ast \subseteq \mathcal{S}$ of size exactly $k$ partitions $U$; this problem is NP-complete~\citep[problem SP2]{garey1979computers}.

\begin{lemma}[Hard core]\label{lem:hardcore}
    There is a polynomial-time algorithm that maps an X3C instance $(U, \mathcal{S})$ with $\cards{U} = 3k$ and $\cards{\mathcal{S}} = M$ to an APS instance $\mathcal{I}$ with $n_0 = M + 6k$ bidders, $m_0 = 6k + M$ items and values in $[0,1]$, such that:
    if $(U,\mathcal{S})$ admits an exact cover, then $\mathcal{I}$ admits an \emph{exact} PNE; and if it does not, then $\mathcal{I}$ admits no $\varepsilon$-PNE for any $\varepsilon \leq \frac{1}{75}$.
\end{lemma}

Before the formal proof, we describe the idea of the construction.
We construct an auction instance with three types of items: \emph{element items} (one per element), \emph{set-safe items} (one per subset), and \emph{threat-safe items}. We also introduce three classes of bidders: \emph{set bidders} (who value the elements in their subset and their specific set-safe item), \emph{element bidders}, and \emph{threat bidders}.

The core intuition relies on the interplay between the safe items and the element items. The instance is constructed in a way that in any $\varepsilon$-PNE, the element and threat bidders push the price of every element item to a high threshold. Consequently, element items must be won by set bidders.

Because a set bidder must bid aggressively to win element items, it is only profitable to do so if she wins \emph{all three} elements of her corresponding subset and otherwise, the high price makes it strictly better to drop her multiplier to $0$ and secure only her uncontested safe item. If two active set bidders overlap on a single element, the loser pays a high price but fails to secure all three items, triggering a profitable deviation to drop out. Therefore, the active set bidders must be completely disjoint. Since all $3k$ elements must be won, these disjoint active set bidders correspond exactly to an exact cover of $U$.

Conversely, if an exact cover exists, the covering set bidders can bid to exactly tie the threat bidders, winning via lexicographic tie-breaking, while non-covering set bidders bid $0$ to safely claim their safe items. We show that this forms an exact PNE.

\begin{proof}
Given an X3C instance $(U, \mathcal{S})$, we build an auction instance in polynomial time that admits an equilibrium if and only if the X3C instance admits an exact cover. 
To facilitate notation, in this proof we will allow the values of the bidders for the items to be arbitrary positive integers, rather than restricting them to $[0,1]$, and we will fix the value of $\varepsilon$ to be less than or equal to $4$. This is without loss of generality, as we can always scale down the valuations and $\varepsilon$ by the same factor to fit them into $[0,1]$.
We will in fact do this in the end to compute the right value of the slack.
The auction contains the following items:
\begin{itemize}
    \item[-] \emph{Element items} $e_1, \dots, e_{3k}$, one per element of $U$.
    \item[-] \emph{Set-safe items} $s_1, \dots, s_M$, one per set $S_i \in \mathcal{S}$.
    \item[-] \emph{Threat-safe items} $t_1, \dots, t_{3k}$, one per element of $U$. 
\end{itemize}
Here and throughout, we identify each element of $U$ with its corresponding element item under the natural bijection, writing $e_j$ for both; the intended meaning is always clear from context.

The auction also contains the following bidders:

\begin{itemize}
    \item[-] \emph{Set bidders} $A_1, \dots, A_M$ where $v_{A_i, e} = 100$ for each $e \in S_i$, and $v_{A_i, s_i} = 300$, for any $i \in [M]$.
    \item[-] \emph{Element bidders} $E_1, \dots, E_{3k}$: $v_{E_j, e_j} = 50$ and $v_{E_j, t_j} = 100$, for any $j \in [3k]$.
    \item[-] \emph{Threat bidders} $T_1, \dots, T_{3k}$: $v_{T_j, e_j} = 50$ (no safe item), for any $j \in [3k]$.
\end{itemize}
All valuations not listed are $0$. Given the naming of the bidders, the lexicographic tie-breaking rule will naturally prioritize set bidders over element bidders over threat bidders, tie-breaking within groups of same types of bidders using their indices.
We begin by some useful observations about the structure of $\varepsilon$-PNEs in this instance.
\begin{observation}\label{obs:undercut}
    Suppose that at a profile $\vec{\alpha}$ the threat bidder $T_j$ wins $e_j$ at price
    $\tau$, and the highest competing bid on $e_j$ is $r < \tau$. Then in any
    $\varepsilon$-PNE, $\tau - r \leq \varepsilon$.
\end{observation}
\begin{proof}
    For any $b \in (r, \tau)$, the deviation in which $T_j$ bids $b$ still wins
    $e_j$ and yields utility $50 - b$, against the current $50 - \tau$. The gain
    $\tau - b$ has supremum $\tau - r$, so $\tau - r > \varepsilon$ would result in a deviation gaining strictly more than $\varepsilon$.
\end{proof}
Now we can prove the following lemma, characterizing the winners of the element items and their prices in any $\varepsilon$-PNE.

\begin{lemma}\label{lem:floor}
    In any $\varepsilon$-PNE, every element $e_j$ is won by a set bidder at a price
    $p_j \geq 50 - \varepsilon$.
\end{lemma}
\begin{proof}
We start by proving the following claim.
\begin{claim}
    In any $\varepsilon$-PNE, the price of each element item is at least $50 - \varepsilon$.
\end{claim}
\begin{proof}
    Let $p \coloneqq \min_{j'} p_{j'}$ be the smallest price over all element items, attained at some element $e_j$; in particular $p_{j'} \ge p$ for every element $e_{j'}$.
    It suffices to show $p \ge 50 - \varepsilon$, since then every element price is at least $50 - \varepsilon$.
    For the sake of contradiction assume that there exists element item $e_j$ with price $p$ such that $p < 50 - \varepsilon$ and let $g$ be its winner.
    We exhibit a deviation gaining strictly more than $\varepsilon$, contradicting the $\varepsilon$-PNE property.

    First, if $g \neq T_j$, then $T_j$ deviating to bid $p + \eta$ (for a small $\eta$) wins $e_j$ for utility $50 - p - \eta$.
    As she currently gets $0$ from $e_j$ and $50 - p > \varepsilon$, her utility increases by $50 - p - \eta > \varepsilon$.

    If $g = T_j$, let $r < p$ be the highest competing bid.
    By \cref{obs:undercut}, $r \geq p - \varepsilon$.
    \begin{itemize}
        \item[-] If $p < \tfrac{50 - \varepsilon}{3}$: the element bidder $E_j$ deviating to bid $p$ for element item $e_j$ wins both $e_j$ (since as mentioned before, ties in the analysis of the best-response benefit the deviator) and her uncontested safe item, for utility $150 - 3p$.
        Her current utility is at most $100$, so she gains at least $50 - 3p > \varepsilon$.
        \item[-] If $p \geq \tfrac{50 - \varepsilon}{3}$: then $r \geq p - \varepsilon
        \geq \tfrac{50 - 4\varepsilon}{3}$, and the runner-up is either $E_j$ or a
        set bidder.
        \begin{itemize}
            \item If the runner-up is $E_j$, she loses $e_j$ while paying $2r$ for her safe
            item. In this case reducing her multiplier to $\alpha_{E_j} = 0$ gains her $2r \geq \tfrac{100 -
            8\varepsilon}{3} > \varepsilon$ gain (since $\varepsilon \leq 4$).
            \item If the runner-up is a set bidder $A_i$ bidding $r$, then $A_i$ bids $r$ on
            each of her element items.
            Any element items that she values positively other than $e_j$ have price at least $p > r$ (minimality), so $A_i$ loses all of them and wins only
            her safe item, and gains utility of $300 - 3r$. However, by reducing her multiplier to $\alpha_{A_i} = 0$ she
            gains $3r \geq 50 - 4\varepsilon > \varepsilon$ (since $\varepsilon \leq 4$).
        \end{itemize}
    \end{itemize}
    Each case contradicts the $\varepsilon$-PNE property, so $p \geq 50 - \varepsilon$.
\end{proof}

    Next, we show that no element is won by a threat or element bidder.
    Fix $e_j$, whose winning bid is $p_j \geq 50 - \varepsilon$, as we just proved.
    If $E_j$ wins $e_j$, then $\alpha_{E_j} \geq \tfrac{50 - \varepsilon}{50}$, so her utility is at most $150 (1- \alpha_{E_j}) \leq 3\varepsilon$ which means by reducing her multiplier to $0$ she gains at least $100 - 3\varepsilon > \varepsilon$ (since $\varepsilon \leq 4$).
    If $T_j$ wins the item, her utility is at most $\varepsilon$, and by
    \cref{obs:undercut} the runner-up bids $r \geq 50 - 2\varepsilon$. If that
    runner-up is $E_j$, she loses $e_j$ while paying at least $100 - 4\varepsilon$
    for her safe item, so she has utility at most $4\varepsilon$ and she can increase her utility by at least $100 - 4\varepsilon > \varepsilon$ by reducing her multiplier to $\alpha_{E_j} = 0$.
    If instead that runner-up is a set bidder $A_i$ bidding $r \geq 50 - 2\varepsilon$, she currently loses $e_j$ and wins at most her safe item and two
    elements, for utility at most $(300 - 3r) + 2(100 - r) = 500 - 5r \leq 250 + 10\varepsilon$; by reducing her multiplier to $\alpha_{A_i} = 0$ she gains at least $50 - 10\varepsilon > \varepsilon$, which holds for $\varepsilon \leq 4$.
    The only remaining option is for the winner to be a set bidder.
\end{proof}

\paragraph{Correctness.}
The reduction is clearly polynomial. We show that the auction instance admits an
    $\varepsilon$-PNE iff $(U, \mathcal{S})$ has an exact cover.

    $\Leftarrow:$ Let $\mathcal{S}^\ast$ be an exact cover. 
    We will call all set bidders $A_i$ with $S_i \in \mathcal{S}^\ast$ \emph{covering} and all other set bidders \emph{non-covering}.
    Consider the
    profile with $\alpha_{A_i} = \tfrac12$ for $S_i \in \mathcal{S}^\ast$,
    $\alpha_{A_i} = 0$ for $S_i \notin \mathcal{S}^\ast$, $\alpha_{T_j} = 1$, and
    $\alpha_{E_j} = 0$ for all $j$; we check it is an exact PNE (hence an
    $\varepsilon$-PNE).
    A covering bidder $A_i$ bids $50$ on each of her three elements, tying the corresponding threat bidders and winning due to the lexicographic tie-breaking, yielding a total utility of $600(1-\alpha_{A_i}) = 300$. 
    By setting $\alpha_{A_i} \geq \tfrac12$, $A_i$ wins all four items for utility $600(1 - \alpha) \leq 300$, and for $\alpha_{A_i} < \tfrac12$ she loses all elements to the threat bidders and just keeps the safe item for
    utility $300(1 - \alpha) \leq 300$, so $\alpha_{A_i} = \tfrac12$ is the optimal action for $A_i$.
    A non-covering set bidder wins her safe item for free at $\alpha = 0$ (utility $300$); winning any of her already-covered elements needs $\alpha \geq \tfrac12$, giving utility of at most $600(1-\alpha) \leq 300$.
    Each threat bidder $T_j$ ties the covering bidder but loses by priority, so her utility is $0$ and, capped at a bid of $50$, she can never win $e_j$.
    Each element bidder $E_j$ wins her safe item for
    free (utility $100$); winning $e_j$ would require outbidding $50$, i.e.\
    $\alpha_{E_j} = 1$, costing $100$ on the safe item for utility $0$.
    Therefore, no bidder has a profitable deviation and this profile is an exact PNE (and thus an $\varepsilon$-PNE).

    $\Rightarrow:$ Let $\vec{\alpha}$ be an $\varepsilon$-PNE.
    By \cref{lem:floor}, every element is won by a set bidder bidding at least $50 - \varepsilon$.
    We will call a set bidder \emph{active} if $\alpha_{A_i} \geq \tfrac{50 - \varepsilon}{100}$, i.e.\ she bids at least $50 - \varepsilon$ on
    her elements. Every element is then won by an active set bidder at equilibrium.

    Suppose there were two active set bidders for some element $e_j$; the one losing it, say $A_q$, bids at least $50 - \varepsilon$ yet wins at most her safe item and two elements, for utility at most $500 (1-\alpha_{A_q}) \leq 250 + 5\varepsilon$.
    Reducing her multiplier to $\alpha_{A_q} = 0$ wins her safe item for free and
    gains at least $50 - 5\varepsilon > \varepsilon$ (since $\varepsilon \leq 4$), contradicting the $\varepsilon$-PNE property. Hence, the active set bidders are pairwise disjoint, so each element item is won by the unique active set bidder containing it. 
    Since an active set bidder bids at least $50 - \varepsilon$ on all three of her elements, she wins all three.
    The active set bidders thus form pairwise-disjoint $3$-sets covering all $3k$ elements: an exact cover of size $k$.
    This concludes the proof of correctness of the reduction from X3C to the $\varepsilon$-PNE problem.

    Finally, scaling every valuation by the inverse of the maximum valuation, i.e., dividing by $300$, yields the correct result for our original setting, where valuations are within $[0,1]$.
    The corresponding slack that we retrieve here is then $\tfrac{\varepsilon}{300}=\tfrac{1}{75}$.
    The instance has $3k+M+3k = 6k+M$ items and $M+6k$ bidders, so the two directions established above are exactly the two cases of the statement, and this concludes our proof of \cref{lem:hardcore}.
\end{proof}

The hard core on its own already proves NP-hardness for every fixed $\varepsilon \leq \frac{1}{75}$. To go beyond a constant, we inflate the gap using \cref{lem:replication}: taking $r$ item-disjoint copies of $\mathcal{I}$ multiplies every deviation gain by exactly $r$, so the replicated instance inherits an exact PNE in the positive case, and admits no $\varepsilon$-PNE for any $\varepsilon \leq \frac{r}{75}$ in the negative one. It has $m = r\,m_0$ items, so the slack we can certify is
\[
    \varepsilon \;=\; \frac{r}{75} \;=\; \frac{1}{75}\cdot\frac{m}{m_0},
\]
and it only remains to choose $r$. Since a polynomial-time reduction can afford any $r$ polynomial in $m_0$, taking $r = \lceil m_0^{(1-\delta)/\delta} \rceil$ yields $\varepsilon = \Theta(m^{1-\delta})$, which is \cref{thm:eps-np-hard}. The same identity marks the limit of this approach: an $\varepsilon$ linear in $m$ would require a hard core of constant size, which \cref{thm:const-m-poly} places in $\mathrm{P}$. We now give the formal proof.

\begin{proof}[Proof of \cref{thm:eps-np-hard}]
    We reduce from X3C. Given $(U,\mathcal{S})$, run the algorithm of \cref{lem:hardcore} to obtain the hard core $\mathcal{I}$, with $m_0 = 6k+M$ items. Write $c \coloneq \frac{1-\delta}{\delta} > 0$, set
    \[
        r \coloneq \left\lceil m_0^{\,c} \right\rceil,
    \]
    and output the instance $\mathcal{I}^{(r,0)}$ of $r$ item-disjoint copies of $\mathcal{I}$, together with the slack $\varepsilon \coloneq \frac{r}{75}$.

    \emph{The reduction runs in polynomial time.} Computing $\mathcal{I}$ is polynomial by \cref{lem:hardcore}, and $m_0$ is linear in the size of $(U,\mathcal{S})$. As $\delta$ is a constant, so is $c$, and hence $r \leq m_0^{\,c}+1$ is polynomial in $m_0$. The instance $\mathcal{I}^{(r,0)}$ has the same $n_0$ bidders and $m = r m_0$ items, and its description consists of $n_0 \cdot r m_0$ values, each a copy of a value of $\mathcal{I}$ and hence of the same bit-length; the whole description therefore has polynomial size and is computed in polynomial time. Finally, $\varepsilon = \frac{r}{75}$ is a rational of polynomial bit-length, and all values of $\mathcal{I}^{(r,0)}$ lie in $[0,1]$, being copies of values of $\mathcal{I}$.

    \emph{Correctness.} Note that $\frac{\varepsilon}{r} = \frac{1}{75}$. If $(U,\mathcal{S})$ admits an exact cover, then by \cref{lem:hardcore} the instance $\mathcal{I}$ admits an exact PNE, which is in particular a $\frac{1}{75}$-PNE of $\mathcal{I}$; by \cref{lem:replication} it is therefore an $\varepsilon$-PNE of $\mathcal{I}^{(r,0)}$. If $(U,\mathcal{S})$ admits no exact cover, then $\mathcal{I}$ admits no $\frac{1}{75}$-PNE, and since by \cref{lem:replication} the $\varepsilon$-PNE of $\mathcal{I}^{(r,0)}$ are exactly the $\frac{\varepsilon}{r}$-PNE of $\mathcal{I}$, the replicated instance admits no $\varepsilon$-PNE either. Hence $\mathcal{I}^{(r,0)}$ admits an $\varepsilon$-PNE if and only if $(U,\mathcal{S})$ admits an exact cover.

    \emph{The slack is $\Theta(m^{1-\delta})$.} By construction $\varepsilon = \frac{r}{75} = \frac{1}{75}\cdot\frac{m}{m_0}$. Since $\delta = \frac{1}{c+1}$ and $\frac{c}{c+1} = 1-\delta$, and since $m_0^{\,c} \leq r \leq m_0^{\,c}+1 \leq 2m_0^{\,c}$, multiplying by $m_0$ gives
    \[
        m_0^{\,c+1} \;\leq\; m = r m_0 \;\leq\; 2\,m_0^{\,c+1},
        \qquad\text{hence}\qquad
        \left(\tfrac{m}{2}\right)^{\delta} \;\leq\; m_0 \;\leq\; m^{\delta}.
    \]
    As $\varepsilon = \frac{1}{75}\cdot\frac{m}{m_0}$ is decreasing in $m_0$ for fixed $m$, the upper bound on $m_0$ gives $\varepsilon \geq \frac{1}{75}\cdot\frac{m}{m^{\delta}} = \frac{1}{75}m^{1-\delta}$, and the lower bound gives
    \[
        \varepsilon \;\leq\; \frac{1}{75}\cdot\frac{m}{(m/2)^{\delta}} \;=\; \frac{2^{\delta}}{75}\,m^{1-\delta} \;\leq\; \frac{2}{75}\,m^{1-\delta},
    \]
    which are the two bounds in the statement. In particular $\varepsilon = \Theta(m^{1-\delta})$, with constants depending only on $\delta$.
\end{proof}

\subsection{Polynomial-time algorithm for constant \texorpdfstring{$m$}{m}}
We now provide an efficient algorithm for the case where the number of items $m$ is constant.
The idea is to write the equilibrium conditions as a linear program, a solution to which would then yield an $\varepsilon$-PNE.
To achieve this, we need to consider the bidders' best-response conditions and make sure that we find $\vec \alpha$ in which each agent is playing her best response up to an $\varepsilon$. To do that we have to compare each bidders' current utility to her utility when deviating to a new multiplier, and since this is a continuous space it cannot be simply done in an LP.

A useful observation that will help us here is that when thinking about the best-responses of bidder $i$ possibly deviating from $\alpha_i $ to $\alpha'_i$, if $\vec \alpha_{-i}$ is known and we start by setting $\alpha'_i = 0$ and keep increasing it to 1, $i$ starts to win the items one by one in a specific order that depends on the valuations and the multipliers of the rest of the bidders. The candidate best-responses of $i$ are values of $\alpha'_i$ in which she starts to win a new item. Another way of thinking about this is that if we knew said ordering of the items and the highest bidder of each item w.r.t. $\vec \alpha_{-i}$, we could compute all values of $\alpha'_i$ that are candidates to be the best-response of $i$ to $\vec \alpha_{-i}$, and also compute bidder $i$'s utility in these values as a linear function of the valuations and the multipliers of the rest of the bidders. Note that since bidder $i$ wins the items one by one when increasing her multiplier from $0$ to $1$, there are at most $m+1$ different candidate values (including $\alpha'_i = 0$). We can simply add the constraint that the utility of bidder $i$ in each of these candidate best-responses should be less than her utility when setting her multiplier to $\alpha_i$. If we do the same for all the bidders, the LP will make sure that $\vec \alpha$ is an equilibrium.

We use this observation to define a configuration of the problem as a choice of the highest bidder and the runner-up of each item alongside an ordering for each bidder. We show that given a configuration we can use the above observation to write an LP to check whether an $\varepsilon$-PNE with such a configuration exists. In addition, we show that with a constant number of items there are only $\mathcal{O}(n^{4m})$ configurations to consider, and we can run a separate LP for each of them and solve the problem in polynomial time.

\begin{theorem}\label{thm:const-m-poly}
    For any constant number of items $m$ and any constant $\varepsilon \ge 0$, the problem of deciding the existence of an $\varepsilon$-PNE in an APS game, and computing one if it exists, is in $\mathrm{P}$.
\end{theorem}

\begin{proof}
    To compute the best-response of bidder $i$ in an APS game, we only need to know the highest (in case $i$ would like to compete for it) or second highest (in case $i$ is already winning but wants to win at a lower price) bid of each item.
    Consequently, in any equilibrium, at most $2m$ bidders actively affect the allocation and the candidate best-responses, while the remaining bidders win nothing.
    We will construct a polynomial-time algorithm that enumerates all possible allocations.
    For each allocation, we will formulate the equilibrium conditions as a Linear Program of polynomial size to the input.

    An item that no bidder values positively is never allocated and never enters a prefix of any bidder, so we discard such items. For every remaining item $j$, let $N_j \coloneq \{ i \in \mathcal{N} : v_{i,j} > 0 \}$ be the set of bidders that value it, and for every bidder $i$, let $M_i \coloneq \{ j \in [m] : v_{i,j} > 0 \}$ be the set of items that she values.
    We begin by enumerating the following two sets, consisting of the winners and the runners-up for each auction:
    \begin{enumerate}
        \item \textbf{The Winners:} A vector $\vec{w} = (w_1, w_2, \dots, w_m)$ assigning a winning bidder $w_j \in N_j$ to each item $j$.
        \item \textbf{The Runners-up:} A vector $\vec{r} = (r_1, r_2, \dots, r_m) \in [n]^m$ assigning a runner-up bidder $r_j$ to each item $j$, such that $r_j \neq w_j$.
    \end{enumerate}
    Let $A = \bigcup_{j=1}^m \{w_j, r_j\}$ be the set of \emph{active bidders}. Note that the size of this set is bounded by $2m$ which is a constant. We set $\alpha_k = 0$ for every \emph{inactive bidder} $k \notin A$; as we show in the correctness argument below, this is without loss of generality.

    To ensure the equilibrium conditions hold for each bidder, we have to compute their utility when deviating to a new multiplier.
    However, it is not straightforward how to write this as a linear constraint to the LP.
    Instead, we will fix the order in which they would win items if they increase their pacing multipliers.
    For an active bidder, we guess this order directly:
    \begin{enumerate}
        \setcounter{enumi}{2}
        \item \textbf{The Threshold Orderings:} For each active bidder $i \in A$, we guess a permutation $\pi_i$ of the items $M_i$. This permutation represents the sequence in which items are secured as bidder $i$'s multiplier $\alpha_i$ increases. Let $\Pi = \{\pi_i\}_{i \in A}$ be the set of these permutations.
    \end{enumerate}
    For the inactive bidders we cannot do the same, since guessing their orders independently would give $(m!)^{\,n - \cards{A}}$ choices, which is exponential in $n$. Instead, we only consider the orders that can occur together.
    The highest bid that an inactive bidder $k$ faces on an item $j \in M_k$ is the one of the winner, so her threshold is $t_{k,j} = \alpha_{w_j} v_{w_j,j} / v_{k,j}$, a linear function of the active multipliers, which range over a space of dimension $\cards{A} \le 2m$. The order of $t_{k,j}$ and $t_{k,j'}$ changes only across the hyperplane $t_{k,j} = t_{k,j'}$, and there are $K = \mathcal{O}(n m^2)$ such hyperplanes.
    An arrangement of $K$ hyperplanes in dimension $d$ has at most $\sum_{i=0}^{d} \binom{K}{i} = \mathcal{O}(K^d)$ full-dimensional regions~\citep{buck1943partition}, which can be listed in time polynomial in $K$ for fixed $d$~\citep{edelsbrunner1987algorithms}; here that is $\mathcal{O}\big((n m^2)^{2m}\big) = \mathcal{O}(n^{2m})$ regions for constant $m$.
    The order of the thresholds of every inactive bidder is the same at all points of a full-dimensional region $C$, and it remains a valid weak order on the closure $\overline{C}$, where some of these thresholds may become equal. Since the closures of the full-dimensional regions cover the whole space, every choice of the active multipliers lies in at least one of them. This gives the last part of the configuration:
    \begin{enumerate}
        \setcounter{enumi}{3}
        \item \textbf{The Inactive Region:} A full-dimensional region $C$ of this arrangement. For each inactive bidder $k$, let $\pi_k$ be the order of her thresholds on the items $M_k$ in $C$.
    \end{enumerate}

    There are $n^m$ ways to choose the winners $\vec{w}$ and at most $n^m$ ways to choose the runners-up $\vec{r}$. For a fixed set $A$ of size at most $2m$, there are $(m!)^{\cards{A}} \le (m!)^{2m}$ ways to choose the threshold permutations $\Pi$, and $\mathcal{O}(n^{2m})$ ways to choose the region $C$. Thus, the total number of configurations to check is $\mathcal{O}(n^{4m} \cdot (m!)^{2m})$. Since $m$ is constant, this is bounded by $\mathcal{O}(n^{4m})$, which is polynomial in $n$. Our algorithm enumerates all such configurations. Formally, a \emph{configuration} is a tuple $(\vec{w}, \vec{r}, \Pi, C)$ that fixes a winner $w_j$ and a runner-up $r_j$ for every item $j \in [m]$, a threshold ordering $\pi_i$ for every active bidder $i \in A$, and a region $C$ that fixes the threshold orderings of the inactive bidders.

    \paragraph{The Linear Program Formulation.}
    We will now proceed to writing the LP for each configuration $(\vec{w}, \vec{r}, \Pi, C)$, in order to find a valid strategy profile $\vec{\alpha}$. The variables of the LP are the multipliers $\alpha_i \in [0,1]$ for the active bidders $i \in A$. The constraints are defined as follows:

\textbf{1. Allocation Constraints:}
    We must enforce that our guess for the winners and runners-up is correct. Under the lexicographic rule~\eqref{eq:lex-winner}, bidder $w_j$ wins item $j$ exactly when no bidder that values the item bids more than her, and every such bidder with a smaller index bids strictly less. A bidder that does not value item $j$ bids $0$ on it and never wins it, so she imposes no condition. For each item $j \in [m]$:
    \begin{align}
        \alpha_{w_j} v_{w_j, j} &\ge \alpha_i v_{i, j} \qquad \forall i \in N_j \setminus \{w_j\}, \text{ strictly if } i < w_j, \label{eq:m_winner} \\
        \alpha_{r_j} v_{r_j, j} &\ge \alpha_i v_{i, j} \qquad \forall i \in A \setminus \{w_j, r_j\}. \label{eq:m_runnerup}
    \end{align}
    Constraint~\eqref{eq:m_winner} ranges over all bidders that value item $j$, including the inactive ones. As an inactive bidder bids $0$, her constraint is trivial unless her index is smaller than $w_j$, in which case it requires the bid of $w_j$ to be positive: otherwise all bids on $j$ are $0$, and the item goes to the inactive bidder. The runner-up only needs to have a highest bid among the bidders other than $w_j$, so~\eqref{eq:m_runnerup} need not be strict, and the inactive bidders satisfy it automatically.

    \textbf{Handling strict inequalities.}
    A linear program cannot impose strict inequalities directly, but strict feasibility is decided by a standard slack-maximization step. We introduce a single variable $\delta \ge 0$, replace every strict requirement $\alpha_a v_{a, j} > \alpha_b v_{b, j}$ by
    \begin{equation*}
        \alpha_a v_{a, j} \ge \alpha_b v_{b, j} + \delta,
    \end{equation*}
    leave the remaining inequalities non-strict, and maximize $\delta$ (capped at $\delta \le 1$ to stay bounded) subject to all constraints of the program. If the optimum has $\delta^\ast > 0$, its maximizer satisfies every strict requirement and is a valid profile for the configuration; if $\delta^\ast = 0$, no profile realizes the configuration and it is discarded. This adds a single variable and a single objective, so the program remains polynomial in size.

    \textbf{2. Deviation Threshold Orderings:}
    For an active bidder $i \in A$, the minimum multiplier required to win item $j$ is determined by the highest competing bid. We define the threshold $t_{i,j}$ as:
    \begin{equation}\label{eq:threshold-def}
        t_{i,j} \coloneq \begin{cases}
            \frac{\alpha_{r_j} v_{r_j, j}}{v_{i, j}} & \text{if } i = w_j \\
            \frac{\alpha_{w_j} v_{w_j, j}}{v_{i, j}} & \text{if } i \neq w_j
        \end{cases}
    \end{equation}
    This is the threshold $t_{i,j}$ of~\eqref{eq:threshold} written out for the guessed winners and runners-up, since the highest bid among the bidders other than $i$ is the one of $r_j$ when $i = w_j$, and the one of $w_j$ otherwise; we only use it for the items $j \in M_i$, since by the convention $t_{i,j} = \infty$ for $v_{i,j} = 0$ an item that $i$ does not value never enters one of her prefixes.
    Notice that $t_{i,j}$ is strictly a linear function of either $\alpha_{w_j}$ or $\alpha_{r_j}$. We enforce the guessed threshold ordering $\pi_i$ via the linear constraints:
    \begin{equation}\label{eq:m_ordering}
        t_{i, \pi_i(1)} \le t_{i, \pi_i(2)} \le \dots \le t_{i, \pi_i(\cards{M_i})} \qquad \forall i \in A.
    \end{equation}
    For the inactive bidders, we require the active multipliers to lie in the closure $\overline{C}$ of the region $C$, which is described by one weak linear inequality per hyperplane of the arrangement. This orders the thresholds of every inactive bidder $k$ weakly according to $\pi_k$.

    \textbf{3. Equilibrium (Best-Response) Conditions:}
    By \cref{lem:deviation-sup}, bidder $i$'s best-response payoff is the largest of the values obtained on the prefixes $S_{i,k} = \{\pi_i(1), \dots, \pi_i(k)\}$, each evaluated at the threshold $\alpha_i' = t_{i, \pi_i(k)}$.
    Let $W_i = \{j \in [m] \mid w_j = i\}$ be the set of items bidder $i$ currently wins; for an inactive bidder, $W_i = \emptyset$ and $\alpha_i = 0$. At an $\varepsilon$-PNE she is best-responding up to $\varepsilon$, so her current utility must be within $\varepsilon$ of each of these prefix values. For every bidder $i \in \mathcal{N}$, active or inactive, and every prefix length $k \in \{1, \dots, \cards{M_i}\}$:
    \begin{equation}\label{eq:m_equilibrium}
        (1 - \alpha_i) \sum_{j \in W_i} v_{i,j} \ge \left(1 - t_{i, \pi_i(k)}\right) \sum_{j \in S_{i,k}} v_{i,j} - \varepsilon
    \end{equation}
    Because $t_{i, \pi_i(k)}$ is linear and $\varepsilon$ is a constant, \eqref{eq:m_equilibrium} is purely a linear inequality.
    If several items share a threshold, the ordering $\pi_i$ separates them arbitrarily, which is harmless: a deviation to that threshold wins bidder $i$ some subset of the tied items, so the prefix containing all of them gives the largest right-hand side in~\eqref{eq:m_equilibrium}, and imposing that constraint implies the ones for every other way the tie could fall.

    \paragraph{Correctness.}
    For each of the $\mathcal{O}(n^{4m})$ configurations, the resulting LP has $\cards{A} + 1 \le 2m + 1$ variables (the multipliers $\alpha_i$ for $i \in A$ and the slack $\delta$) and $\mathcal{O}(n m^2)$ linear constraints (\eqref{eq:m_winner} through \eqref{eq:m_equilibrium}, the inequalities describing $\overline{C}$, and the bounds $0 \le \alpha_i \le 1$). Because $m$ is constant, solving this LP takes polynomial time.

    If the LP for some configuration has an optimum with $\delta^\ast > 0$, the resulting $\vec{\alpha}$ (padded with $0$s for bidders outside $A$) constitutes an $\varepsilon$-PNE. By~\eqref{eq:m_winner} and~\eqref{eq:m_runnerup}, every item $j$ is won by $w_j$ and $r_j$ has a highest bid among the other bidders, so the thresholds in~\eqref{eq:threshold-def}, and the ones of the inactive bidders, are the thresholds of~\eqref{eq:threshold}. These are ordered according to the orderings $\pi_i$, and by \cref{lem:deviation-sup} and~\eqref{eq:m_equilibrium} every bidder, active or inactive, is best-responding up to $\varepsilon$. Setting $\varepsilon = 0$ recovers the exact-PNE case.

    Conversely, suppose an $\varepsilon$-PNE $\vec{\alpha}$ exists, and let $w_j$ be the winner of each item $j$ and $r_j$ a bidder with the highest bid on $j$ among the others. We may assume that $\alpha_k = 0$ for every bidder $k$ that is neither a winner nor a runner-up of any item, since setting all these multipliers to $0$ yields another $\varepsilon$-PNE. Indeed, on every item such a bidder bids at most as much as the runner-up, so lowering her bid changes no winner and no price, and hence no bidder's current utility. The thresholds of every other bidder $i$, which depend only on $\max_{k' \neq i} \alpha_{k'} v_{k',j}$, are attained by the winner or the runner-up of each item and are unchanged as well, so by \cref{lem:deviation-sup} so is her best-response payoff. Finally, bidder $k$ herself has utility $0$ before and after the change, while her best-response payoff $u_k^{\ast}(\vec{\alpha_{-k}})$ does not depend on her own multiplier.
    The profile then induces winners $\vec{w}$ and runners-up $\vec{r}$, and we let $\pi_i$ be the order of the thresholds of each active bidder $i$. The active multipliers lie in the closure of at least one full-dimensional region $C$ of the arrangement, and this fixes the orderings of the inactive bidders. In the resulting configuration, $\vec{\alpha}$ satisfies~\eqref{eq:m_winner}, with the strict inequalities holding strictly since $w_j$ is the actual winner of $j$, as well as~\eqref{eq:m_runnerup},~\eqref{eq:m_ordering}, and the constraints of $\overline{C}$. It also satisfies~\eqref{eq:m_equilibrium}: by \cref{lem:deviation-sup} each right-hand side is at most the best-response payoff of the bidder minus $\varepsilon$, and the bidder is best-responding up to $\varepsilon$. Therefore, the optimum of this LP has $\delta^\ast > 0$.

    The procedure therefore runs in time $n^{\mathcal{O}(m)}$, polynomial for constant $m$, and accepts if and only if an $\varepsilon$-PNE exists; whenever it accepts, the feasible $\vec{\alpha}$ found for the accepting configuration (padded with $0$s outside $A$) is an explicit $\varepsilon$-PNE.
\end{proof}

\subsection{Polynomial-time algorithm for constant \texorpdfstring{$n$}{n}}

We will now provide an efficient algorithm for the case where the number of bidders $n$ is a constant.
We begin with the special case of two bidders, which serves as a warm-up: it already exhibits the prefix structure of best-responses on which the general algorithm relies.
We present this case first to convey the intuition, and then turn to a general (constant) number of bidders.

\begin{theorem}
    For any constant $\varepsilon \ge 0$, the problem of deciding the existence of an $\varepsilon$-PNE in an APS game with $2$ bidders, and computing one if it exists, is in $\mathrm{P}$.
\end{theorem}
\begin{proof}
    For each item $j$, let $r_j\coloneq \frac{v_{2,j}}{v_{1,j}}$, with the conventions that $r_j = 0$ if $v_{2,j} = 0$ and $r_j = \infty$ if $v_{1,j} = 0$; items that neither bidder values are never allocated and can be discarded.
    Note that, under our tie-breaking rule, an item valued by a single bidder is won by her regardless of the multipliers; such items sit at the two ends of the ordering below, leaving the prefix structure unaffected.
    Since the auctions are independent and run in parallel, the order in which we consider the items does not matter.
    Therefore, assume without loss of generality that the items are in increasing order of $r$.
    This means that, fixing $(\alpha_1,\alpha_2)$, the resulting allocation can be viewed as a prefix of items $S_k=\{1,2,\ldots,k\}$ being allocated to bidder $1$, and the rest of them $M \setminus S_k = \{k+1,\ldots,m\}$ to bidder $2$, for some $k \in \{0,1,\ldots, m\}$.

    Let $D$ be the set of items valued by both bidders. For a split after item $k$ to occur, bidder $1$ must win every item of $S_k$ and bidder $2$ every item of $M \setminus S_k$. This is impossible if $S_k$ contains an item valued only by bidder $2$ or $M \setminus S_k$ contains an item valued only by bidder $1$, so we only consider the values of $k$ for which neither happens; for the remaining items it holds automatically. For the items in $D$, since bidder $1$ wins ties, the split requires
    \begin{equation}\label{eq:ratio}
        \alpha_1 v_{1,j} \geq \alpha_2 v_{2,j} \quad \text{for } j \in S_k \cap D,
        \qquad
        \alpha_2 v_{2,j} > \alpha_1 v_{1,j} \quad \text{for } j \in D \setminus S_k.
    \end{equation}
    Written in this multiplicative form, the condition involves no division, and it remains correct when $\alpha_2 = 0$, in which case bidder $2$ wins no item of $D$. The split results in the following bidder utilities:
    \begin{align*}
       u_1(\alpha_1,\alpha_2) &= (1-\alpha_1) V_1(S_k) \\
       u_2(\alpha_1,\alpha_2) &= (1-\alpha_2) V_2(M \setminus S_k)
    \end{align*}
    where $V_i(X) = \sum_{j \in X} v_{i,j}$.

    We will iterate over all $m+1$ possible splits and check whether there are $(\alpha_1,\alpha_2)$ consistent with each split that satisfy the equilibrium conditions.

    First, we will look at bidder $1$'s deviations.
    Her threshold on an item $j$ she values is $t_{1,j} = \frac{\alpha_2 v_{2,j}}{v_{1,j}} = \alpha_2 r_j$, and $t_{1,j} = \infty$ on an item she does not value. Since each $r_j$ is a constant, these thresholds are linear in $\alpha_2$, and they are ordered by $r$ regardless of the multipliers, so the prefixes of \cref{lem:deviation-sup} for bidder $1$ are the sets $S_\ell$.
    By \cref{lem:deviation-sup}, the best she can obtain from a deviation securing $S_\ell$, for $\ell \geq 1$ with $r_\ell < \infty$, is the value at $\alpha_1' = \alpha_2 r_\ell$, namely $(1-\alpha_2 r_\ell)V_1(S_\ell)$.
    Thus, at an $\varepsilon$-PNE, for any such $\ell$ the following must hold:
    \begin{equation}\label{eq:1BR}
        (1-\alpha_1)V_1(S_k) \geq (1-\alpha_2r_\ell)V_1(S_\ell) - \varepsilon
    \end{equation}

    The second bidder is treated in exactly the same way, the only difference being that she loses, rather than wins, the tie she creates.
    Her threshold on an item $j$ she values is $t_{2,j} = \frac{\alpha_1 v_{1,j}}{v_{2,j}}$, and $t_{2,j} = \infty$ otherwise; these thresholds are ordered by \emph{decreasing} $r$, so her prefixes are the suffixes $M \setminus S_\ell$, and the largest threshold in $M \setminus S_\ell$ is the one of item $\ell+1$.
    For $\ell \leq m-1$ with $v_{2,\ell+1} > 0$, her cheapest deviation securing $M \setminus S_\ell$ is to set the multiplier to $\alpha_2' = \frac{\alpha_1 v_{1,\ell+1}}{v_{2,\ell+1}}$ and tie bidder $1$ on item $\ell+1$, and by \cref{lem:deviation-sup} the best she can obtain from such a deviation is again the value at that multiplier, whether or not it is attained.
    Thus, at an $\varepsilon$-PNE, for any such $\ell$:
    \begin{equation}\label{eq:2BR}
        (1-\alpha_2) V_2(M \setminus S_k) \geq \left( 1 - \frac{\alpha_1 v_{1,\ell+1}}{v_{2,\ell+1}}  \right) V_2(M \setminus S_\ell) - \varepsilon
    \end{equation}
    The remaining deviations, to the empty prefix for bidder $1$ and the empty suffix for bidder $2$, have value $0$ and impose no constraint. If several items share the same ratio, the prefix (or suffix) containing all of them gives the largest value, so the constraints for the prefixes separating them are implied.

    To conclude the proof, notice that we can first compute the $r_j$ values and sort the items in $O(m\log m)$ time and then, for each admissible $k \in \{0,1,\ldots, m\}$, write down the system consisting of at most $2m$ inequalities from \eqref{eq:1BR} and \eqref{eq:2BR} and at most $m$ from \eqref{eq:ratio}, all of them linear in $(\alpha_1,\alpha_2)$.
    As in the proofs of \cref{thm:const-m-poly,thm:const-n-poly}, we handle the strict inequalities of \eqref{eq:ratio} with a single slack variable $\delta$: we replace each of them by $\alpha_2 v_{2,j} \geq \alpha_1 v_{1,j} + \delta$ and maximize $\delta \leq 1$, so that the split is realizable exactly when the optimum is positive; this is a linear program, solved in polynomial time.
    If none of the constructed systems is feasible, then the instance has no $\varepsilon$-PNE; otherwise, solving any one feasible system yields a pair $(\alpha_1,\alpha_2)$ that is itself an $\varepsilon$-PNE, so the algorithm both decides existence and computes an equilibrium whenever one exists.
\end{proof}

We now turn to the case of a constant number of bidders $n$. The roles of $m$ and $n$ are reversed relative to \cref{thm:const-m-poly}: since the number of items $m$ may now be large, we can no longer enumerate allocations directly. Instead we exploit that the strategy space $[0,1]^n$ is low-dimensional. The pairwise comparisons that govern the outcome, who outbids whom on each item and the order in which each bidder would acquire the items, are all linear in the multipliers, so they carve $[0,1]^n$ into a hyperplane arrangement whose cells each fix the entire allocation and best-response structure. For constant $n$ this arrangement has only polynomially many cells, and on each of them we solve a single linear program for an $\varepsilon$-PNE consistent with that cell.

\begin{theorem}\label{thm:const-n-poly}
    For any constant number of bidders $n$ and any constant $\varepsilon \ge 0$, the problem of deciding the existence of an $\varepsilon$-PNE in an APS game, and computing one if it exists, is in $\mathrm{P}$.
\end{theorem}

\begin{proof}
    Because $m$ is not constant, we cannot simply enumerate all $n^m$ possible item allocations. However, the entire auction outcome and the bidders' strategic incentives are driven entirely by a set of simple pairwise comparisons between the bidders' pacing multipliers similar to what we saw in the case of two agents.

    By identifying all the linear boundaries where these comparisons change, we can partition the continuous space of all possible strategies $[0,1]^n$ into a polynomial number of distinct regions. Within any single region, the allocation of the items is completely locked, allowing us to formulate the equilibrium conditions as a linear program.

    \paragraph{The Linear Boundaries.}
    Two types of structural events dictate the auction's mechanics:
    \begin{enumerate}
        \item \textbf{Allocation Boundaries:} To determine who wins item $j$, we compare bids. Bidder $a$ beats bidder $b$ on item $j$ if $\alpha_a v_{a,j} > \alpha_b v_{b,j}$. The boundary where this swaps is the linear equation:
        \begin{equation}\label{eq:boundary-alloc}
            \alpha_a v_{a,j} - \alpha_b v_{b,j} = 0
        \end{equation}
        There are at most $\binom{n}{2}m$ such boundaries.
        
        \item \textbf{Deviation Threshold Boundaries:} To secure item $j$, bidder $i$ must beat the highest bid of other agents, and who that competitor is depends on whether $i$ already holds the item: define $k_i(j) \coloneq r_j$ if $i = w_j$ and $k_i(j) \coloneq w_j$ if $i \neq w_j$, exactly as in the case split of \eqref{eq:threshold-def}. The threshold multiplier bidder $i$ needs to win the item is then $t_{i,j} = \alpha_{k_i(j)} v_{k_i(j),j} / v_{i,j}$, which is again the threshold of~\eqref{eq:threshold}, with the same convention $t_{i,j} = \infty$ when $v_{i,j} = 0$.
        Bidder $i$'s best-response requires sorting these thresholds to find the optimal prefix of items to target. The order of two thresholds changes when $t_{i,j} = t_{i,j'}$, which forms the linear boundary:
        \begin{equation}\label{eq:boundary-thresh}
            \alpha_a \frac{v_{a,j}}{v_{i,j}} - \alpha_b \frac{v_{b,j'}}{v_{i,j'}} = 0
        \end{equation}
        Considering all bidders $i$, all pairs of items $j, j'$, and all potential competitors $a, b$, there are at most $n \binom{m}{2} n^2$ such boundaries; since $a$ and $b$ already range over all bidders, this count already covers both cases $k_i(j) = r_j$ and $k_i(j) = w_j$ without any adjustment.
    \end{enumerate}

    In total, there are $K = \mathcal{O}(n^3 m^2)$ linear boundaries of the form $\alpha_x - C\alpha_y = 0$ passing through the origin in $\mathbb{R}^n$.

    \paragraph{The Enumeration Space.}
    Each of the $K$ linear equations defines a hyperplane through the origin in $\mathbb{R}^n$, and together they form a \emph{hyperplane arrangement}. This arrangement subdivides $\mathbb{R}^n$ into \emph{faces} of every dimension, and within the relative interior of a face the sign, either $<$, $=$, or $>$, of each of the $K$ equations is constant. We call such a sign vector in $\{<, =, >\}^K$ a \emph{configuration}, and our algorithm enumerates all realizable ones.

    We must account for faces of \emph{all} dimensions, not only the full-dimensional regions carved out by strict inequalities. An equality sign records a tie between two paced bids, which the lexicographic rule resolves in favour of the lower-indexed bidder among those that value the item; an equilibrium may therefore sit exactly on such a tie, that is, on a lower-dimensional face of the arrangement, and the algorithm must enumerate these boundary faces as well. The full-dimensional regions on their own number at most $\sum_{i=0}^{n}\binom{K}{i}$~\citep{buck1943partition}, which is already polynomial in $K$ for constant $n$: since $\binom{K}{i}=\tfrac{K(K-1)\cdots(K-i+1)}{i!}\le K^{i}$,
    \begin{equation}\label{eq:cell-count}
        \sum_{i=0}^{n}\binom{K}{i} \;\le\; \sum_{i=0}^{n}K^{i} \;\le\; (n+1)K^{n} \;=\; \mathcal{O}(K^{n}),
    \end{equation}
    the last step absorbing the $n+1$ terms into the constant, which is valid because $n$ is fixed. The lower-dimensional faces, on which one or more equations vanish, could a priori be far more numerous; however, for fixed $n$ the total number of faces of \emph{all} dimensions in an arrangement of $K$ hyperplanes is likewise $\mathcal{O}(K^{n})$~\citep{edelsbrunner1987algorithms}. Substituting $K = \mathcal{O}(n^3 m^2)$, the number of configurations our algorithm enumerates is at most $\mathcal{O}\big((n^3 m^2)^n\big) = \mathcal{O}(m^{2n})$, which is polynomial in $m$ for constant $n$.

    \paragraph{Formulating the Linear Program.}
    For every realizable configuration, the constant signs of our $K$ equations explicitly dictate the complete structure of the auction. Specifically, they provide:
    \begin{itemize}
        \item A fixed winner $w_j$ and runner-up $r_j$ for every item $j$.
        \item A fixed permutation $\pi_i$ for every bidder $i$, representing the sorted order of their deviation thresholds: $t_{i,\pi_i(1)} \le t_{i,\pi_i(2)} \le \dots \le t_{i,\pi_i(m)}$.
    \end{itemize}
    Given this locked structure, we formulate a Linear Program to find a specific pacing multiplier profile $\vec{\alpha} = (\alpha_1, \dots, \alpha_n)$ inside this region that forms an $\varepsilon$-PNE. The variables are $\alpha_i \in [0,1]$ for all $i \in [n]$.

    \textbf{1. Region Constraints:}
    We must ensure the $\vec{\alpha}$ we find actually lives inside the guessed configuration. For every one of the $K$ boundaries, we enforce the sign dictated by the configuration. For example, if the configuration dictates that bidder $a$ beats bidder $b$ on item $j$, we add the linear constraint:
    \begin{equation}
        \alpha_a v_{a,j} \ge \alpha_b v_{b,j}
    \end{equation}
    When the configuration dictates a \emph{strict} sign (one bidder strictly outbidding another, as opposed to a tie that the lexicographic rule resolves in favour of the lower-indexed bidder among those that value the item) the corresponding inequality must be strict. A linear program cannot impose strict inequalities directly, so we decide strict feasibility by maximizing a slack: we introduce a single variable $\delta \ge 0$, replace each strict constraint $\alpha_a v_{a,j} > \alpha_b v_{b,j}$ by $\alpha_a v_{a,j} \ge \alpha_b v_{b,j} + \delta$, keep the equality signs (which encode lexicographic ties) as equalities, and maximize $\delta$ subject to all constraints, capped at $\delta \le 1$. The configuration is realizable exactly when the optimum satisfies $\delta^\ast > 0$.

    \textbf{2. Equilibrium (Best-Response) Conditions:}
    With the thresholds ordered by $\pi_i$, and ties among them treated as in the proof of \cref{thm:const-m-poly}, \cref{lem:deviation-sup} gives bidder $i$'s best-response payoff as the largest of the values obtained on the prefixes $S_{i,k} = \{\pi_i(1), \dots, \pi_i(k)\}$, each evaluated at the threshold $\alpha_i' = t_{i, \pi_i(k)}$. 
    Let $W_i = \{j \in [m] \mid w_j = i\}$ be the set of items bidder $i$ currently wins in this configuration. At an $\varepsilon$-PNE bidder $i$ is best-responding up to $\varepsilon$, so her current utility must be within $\varepsilon$ of each of these prefix values. For every bidder $i \in [n]$ and every prefix length $k \in \{0, 1, \dots, m\}$:
    \begin{equation}
        (1 - \alpha_i) \sum_{j \in W_i} v_{i,j} \ge \left(1 - t_{i, \pi_i(k)}\right) \sum_{j \in S_{i,k}} v_{i,j} - \varepsilon
    \end{equation}
    Because $t_{i, \pi_i(k)} = \alpha_{k_i(\pi_i(k))} v_{k_i(\pi_i(k)), \pi_i(k)} / v_{i, \pi_i(k)}$ is just a constant scalar multiplied by the variable $\alpha_{k_i(\pi_i(k))}$ of the relevant competitor $k_i(\pi_i(k))$ (the runner-up for items $i$ already wins, the current winner otherwise), this is a purely linear inequality.

    \paragraph{Conclusion.}
    We check $\mathcal{O}(m^{2n})$ configurations. For each, we solve an LP with $n + 1$ variables (the multipliers $\alpha_i$ and the tie-breaking slack $\delta$) and $\mathcal{O}(K + nm) = \mathcal{O}(n^3 m^2)$ linear constraints. Since $n$ is constant, solving this LP takes polynomial time. If any configuration's LP is feasible, any valid solution $\vec{\alpha}$ is an $\varepsilon$-PNE, and solving that one LP directly returns such a profile. If all LPs are infeasible, no $\varepsilon$-PNE exists. Therefore, the problem of deciding the existence of an $\varepsilon$-PNE, and computing one if it exists, is in $\mathrm{P}$.
\end{proof}

\section{Inefficiency of Equilibria}\label{sec:inefficiency}
In this section we evaluate the social welfare loss at equilibrium, and our goal is to capture how much social welfare is lost when bidders act strategically rather than optimally, which we measure using the Price of Anarchy (PoA) and Price of Stability (PoS) notions. The \emph{Price of Anarchy} (PoA) compares the optimal welfare to that of the \emph{worst} equilibrium while the \emph{Price of Stability} (PoS) compares it to the \emph{best} one, capturing the loss that is unavoidable even when a benevolent coordinator is free to select the equilibrium. We start with formal definitions of these notions.

\paragraph{Efficiency.}
We will also be interested in the efficiency of the output with respect to the social welfare objective. For a strategy profile $\vec \alpha$ the social welfare of $\vec \alpha$ is defined as:
\begin{equation}
    \sw(\vec{\alpha}) \coloneq \sum_{i \in \mathcal{N}}\sum_{j = 1}^m x_{i,j} \cdot v_{i,j} \cdot
\end{equation}

The \emph{optimal social welfare} of an instance $\mathcal{I}$ is defined as the maximum social welfare achievable over all strategy profiles, i.e., 
$
    \opt(\mathcal{I}) \coloneq \max_{\vec{\alpha}} \sw(\vec \alpha).
$
Notice that in our case the optimal social welfare is trivially achieved by setting each bidder's multiplier to $1$, resulting in each item going to the bidder with the highest value for it, extracting therefore the maximum value possible. Therefore, we can equivalently rewrite the optimum as:
$$
\opt(\mathcal{I}) = \sum_{j=1}^m \max_{i \in \mathcal{N}} v_{i,j}\cdot
$$

\begin{definition}[Price of Anarchy]
The \emph{Price of Anarchy} (PoA) of an instance with respect to the social welfare objective is the worst-case ratio of the optimal social welfare to the social welfare achieved at a pure Nash equilibrium, over all instances:
\begin{equation}
    \mathrm{PoA} \coloneq \sup_{\mathcal{I}:\mathcal{I}\text{ has a PNE}} \;\frac{\opt(\mathcal{I})}{\displaystyle\inf_{\vec{\alpha}\,\in\,\mathsf{NE}(\mathcal{I})} \sw(\vec{\alpha})}\cdot
\end{equation}

\noindent A PoA of $\rho \geq 1$ means that every equilibrium achieves at least a $1/\rho$ fraction of the optimal social welfare.
Notice that the $\sup$ is taken over instances that admit a PNE, since there are instances without any PNE, for which it does not make sense to study the PoA.
\end{definition}

\begin{definition}[Price of Stability]
The \emph{Price of Stability} with respect to the social welfare objective is
\begin{equation}
    \mathrm{PoS} \coloneq \sup_{\mathcal{I}:\mathcal{I}\text{ has a PNE}} \;\frac{\opt(\mathcal{I})}{\displaystyle\sup_{\vec{\alpha}\,\in\,\mathsf{NE}(\mathcal{I})} \sw(\vec{\alpha})},
\end{equation}
where the inner supremum is over all pure Nash equilibria of $\mathcal{I}$, and the outer one over all instances admitting at least one such equilibrium.
\end{definition}

Our main contribution is to show that, for APS games, the two measures coincide. By definition it follows that PoS is at most as much as PoA. Instead of proving separate bounds for these notions we exploit this observation and prove an upper bound on the PoA and a matching lower bound on the PoS.
Our proof of the upper bound has a similar flavour to the proof of \citet{syrgkanis2013composable} that simultaneous first-price auctions are $(1-1/e)$-efficient. Against an arbitrary equilibrium we test each bidder with a single \emph{randomized} deviation, drawn from a fixed distribution on $[0,1-1/e]$ that uniformly scales her values; its expected gain already certifies that the equilibrium captures at least a $1-1/e$ fraction of the optimal welfare.

\begin{theorem}[PoA upper bound]\label{thm:poa-upper}
    In an APS game, the Price of Anarchy with respect to social welfare is at most $\frac{e}{e-1}$.
\end{theorem}

\begin{proof}
    Fix any PNE $\alpha$ and let $\alpha^*$ be a welfare-maximizing profile.
    Let $\opt_i \subseteq [m]$ denote the set of items bidder $i$ is allocated under $\alpha^*$.
    Notice that the sets $\{\opt_i\}_{i\in [n]}$ are pairwise disjoint and denote $\opt = \sum_{i \in [n]} V_i(\opt_i)$, where we have defined $V_i(S) = \sum_{j \in S} v_{i,j}$.
    Let the price of an item $j$ at equilibrium $\alpha$ be $p_j = \max_{k \in [n]} \alpha_k v_{k,j}$, so that the winner of $j$ pays $p_j$ and the \emph{revenue} collected at $\alpha$ is
    \begin{equation*}
        \rev(\alpha) \coloneq \sum_{j=1}^m p_j.
    \end{equation*}
    Since each bidder pays her own bid exactly on the items she wins, the social welfare splits into the bidders' utilities and the revenue:
    \begin{equation}\label{eq:sw-decomp}
        \sw(\alpha) = \sum_{i \in [n]} u_i(\alpha) + \rev(\alpha),
    \end{equation}
    because $\sum_{i \in [n]} u_i(\alpha) = \sum_{j=1}^m \big(v_{w_j,j} - p_j\big) = \sw(\alpha) - \rev(\alpha)$, where $w_j$ denotes the winner of item $j$ at $\alpha$.

    We will consider a \emph{randomized deviation}. This is only an analytic device and does not alter the solution concept: at a pure Nash equilibrium no \emph{pure} deviation is profitable, so the inequality $u_i(\alpha) \ge u_i(\alpha_i', \alpha_{-i})$ holds for every pure $\alpha_i' \in [0,1]$ simultaneously. Averaging it over any distribution of pure deviations therefore preserves it, since the expected utility of the randomized deviation is just a convex combination of pure-deviation utilities, each bounded by $u_i(\alpha)$.
    Let $\mathcal{D}$ be the distribution on $[0,1-1/e]$ with p.d.f. given by:
    \begin{equation*}
        f(t) = \frac{1}{1-t}.
    \end{equation*}
    To see that this is a valid probability distribution, note that
    \begin{equation*}
        \int_0^{1-1/e}f(t) \,\mathrm{d}t = -\ln(1-t) \bigg|_0^{1-1/e} = -\ln(1/e) = 1.
    \end{equation*}
    Since $\alpha$ is a PNE, for any bidder $i$ and any deviation to a pure strategy $\alpha_i'$ it must be that $u_i(\alpha)\geq u_i(\alpha_i',\alpha_{-i})$.
    Taking the expectation of both sides under $T \sim \mathcal{D}$:
    \begin{equation}\label{eq:ub-exp}
        u_i(\alpha) \geq \mathbb{E}_{T\sim \mathcal{D}}[u_i(T,\alpha_{-i})].
    \end{equation}
    Now, fix a bidder $i$ and an item $j \in \opt_i$.
    When $i$ deviates to $\alpha_i'=T$, her bid for $j$ becomes $Tv_{i,j}$ and she wins the item if and only if $Tv_{i,j}\geq p_j$, gaining utility of $(1-T)v_{i,j}$ from that item.
    Therefore, the expected utility gain of $i$'s deviation from item $j$ is
    \begin{equation*}
        \mathbb{E}_{T \sim \mathcal{D}}[(1-T)v_{i,j} \mathbf{1}\{Tv_{i,j}\geq p_j\}].
    \end{equation*}
    As we seek to lower bound this quantity, we will distinguish two cases, based on the value of $\beta_j \coloneq \frac{p_j}{v_{i,j}}$:
    \begin{itemize}
        \item[-] If $\beta_j\leq 1-1/e$, the expectation evaluates to:
        \begin{equation*}
            v_{i,j}\int_{\beta_j}^{1-1/e}(1-t)f(t) \, \mathrm{d}t = v_{i,j}\int_{\beta_j}^{1-1/e}1 \, \mathrm{d}t = v_{i,j}\left(1-\frac{1}{e}-\beta_j \right) = \left( 1- \frac{1}{e}\right) v_{i,j} - p_j
        \end{equation*}
        \item[-] If $\beta_j>1-1/e$, then the indicator is $0$ on the support of $\mathcal{D}$, so the expectation is $0$, while the expression $\left( 1- \frac{1}{e}\right) v_{i,j} - p_j$ is negative.
    \end{itemize}
    Therefore, we can bound the expectation in both cases as:
    \begin{equation}\label{eq:ub-per-item-bound}
        \mathbb{E}_{T \sim \mathcal{D}}[(1-T)v_{i,j} \mathbf{1}\{Tv_{i,j}\geq p_j\}] \geq \left( 1- \frac{1}{e}\right) v_{i,j} - p_j.
    \end{equation}
    Using the fact that $i$'s utility after deviating is lower bounded by her utility contributions from items in $\opt_i$, we get:
    \begin{equation*}
        \mathbb{E}_{T\sim \mathcal{D}}[u_i(T,\alpha_{-i})] \geq \sum_{j\in \opt_i} \mathbb{E}_{T \sim \mathcal{D}}[(1-T)v_{i,j} \mathbf{1}\{Tv_{i,j}\geq p_j\}] \overset{\eqref{eq:ub-per-item-bound}}{\geq} \left( 1-\frac{1}{e}\right)V_i(\opt_i) - \sum_{j \in \opt_i}p_j.
    \end{equation*}
    Combining this with \eqref{eq:ub-exp} and summing over all bidders:
    \begin{equation*}
        \sum_{i \in [n]}u_i(\alpha) \geq \left(1-\frac{1}{e}\right) \sum_{i \in [n]}V_i(\opt_i) - \sum_{i \in [n]}\sum_{j\in\opt_i}p_j = \left(1-\frac{1}{e}\right) \opt -\rev(\alpha),
    \end{equation*}
    where the last equality uses $\sum_{i \in [n]}\sum_{j\in\opt_i}p_j = \sum_{j=1}^m p_j = \rev(\alpha)$, since the welfare-optimal allocation $\{\opt_i\}_{i\in[n]}$ partitions the set of items.
    Finally, substituting the welfare decomposition~\eqref{eq:sw-decomp}, we obtain
    \begin{equation*}
        \sw(\alpha) = \sum_{i \in [n]} u_i(\alpha) + \rev(\alpha) \geq \left(1-\frac{1}{e}\right) \opt,
    \end{equation*}
    which directly gives the upper bound on the PoA:
    \begin{equation*}
        \text{PoA} = \frac{\opt}{\sw(\alpha)} \leq \frac{e}{e-1}
    \end{equation*}
\end{proof}

Having shown that no equilibrium can be too inefficient, we now establish a matching lower bound, which pins down the Price of Stability exactly.

\begin{theorem}[PoS lower bound]\label{thm:pos-lower}
    In an APS game, the Price of Stability with respect to social welfare is at least $\frac{e}{e-1}$.
\end{theorem}
\begin{proof}
    We exhibit a family of instances in which a single \emph{big} bidder, who values every item, faces a pair of identical competitors on each of several \emph{contested} items and \emph{strictly} prefers to concede all of them; pairing the two competitors on each contested item pins its price and keeps an (inefficient) equilibrium alive. The family has three parameters, the numbers $N$ and $M$ of free and contested items and a small perturbation $\eta > 0$ to the competitors' values, and tuning them drives its Price of Stability up to $\frac{e}{e-1}$, though no single instance attains the bound.

Fix integers $N, M \geq 1$ and a real $\eta \in \left(0,\, \tfrac{N}{N+M}\right)$. The instance $\mathcal{I}_{N,M,\eta}$ has:
\begin{itemize}
    \item[-] \emph{Free} items $f_1, \dots, f_N$ and \emph{contested} items $g_1, \dots, g_M$;
    \item[-] a \emph{big} bidder $0$ with $v_{0,f_i} = 1$ and $v_{0,g_k} = 1$ for all $i \in [N]$, $k \in [M]$;
    \item[-] for each $k \in [M]$, a \emph{competitor pair} $A_k, B_k$ with $v_{A_k,g_k} = v_{B_k,g_k} = c_k + \eta$, where $c_k \coloneq \tfrac{k}{N+k}$, and value $0$ on every other item.
\end{itemize}
Ties on any item are broken in favour of the big bidder $0$, and within a competitor pair in favour of $A_k$. The constraint $\eta < \tfrac{N}{N+M} = 1 - c_M$ guarantees $c_k + \eta < 1$ for all $k$, so the welfare-optimal allocation gives every item to bidder $0$, yielding $\opt = N + M$.

\begin{claim}\label{lem:pos-existence}
    The profile $\alpha_0 = 0$ and $\alpha_{A_k} = \alpha_{B_k} = 1$ for all $k \in [M]$ is an exact PNE of $\mathcal{I}_{N,M,\eta}$.
\end{claim}
\begin{proof}
    Bidder $0$ bids $0$, winning all $N$ free items at price $0$ and no contested item, for utility $N$. Any deviation for her that wins her the contested prefix $g_1, \dots, g_\ell$ requires a multiplier of at least $c_\ell + \eta$ (to reach the competitors' bids), giving utility
    \[
        (N+\ell)\bigl(1 - c_\ell - \eta\bigr) = (N+\ell)(1-c_\ell) - (N+\ell)\eta = N - (N+\ell)\eta < N,
    \]
    where we used $(N+\ell)(1-c_\ell) = (N+\ell)\cdot\tfrac{N}{N+\ell} = N$. Hence $\alpha_0 = 0$ is strictly optimal for bidder $0$. Each $A_k$ wins $g_k$ at price $c_k + \eta$, for utility $0$; lowering her multiplier gives $g_k$ to $B_k$ (still utility $0$), and she cannot bid above $1$. Symmetrically, $B_k$ obtains utility $0$ and cannot profitably win. No bidder has a profitable deviation.
\end{proof}
Next, we show that at any PNE of this instance, the big bidder will not win any of the contested items.

\begin{claim}\label{lem:pos-concede}
    In every PNE of $\mathcal{I}_{N,M,\eta}$, bidder $0$ does not win any contested item.
\end{claim}
\begin{proof}
    Let $\alpha$ be a PNE, and suppose for contradiction that bidder $0$ wins a non-empty set $S$ of contested items, with $\ell \coloneq |S| \geq 1$ and largest index $k^\ast \coloneq \max\{k : g_k \in S\} \geq \ell$, so that $c_{k^\ast} \geq c_\ell$ by monotonicity of $c_k$.

    Fix any $g_k \in S$; bidder $0$ wins it at price equal to her own bid $\alpha_0$. If $\alpha_0 < c_k + \eta$, then competitor $A_k$ has a profitable deviation to a multiplier in $\bigl(\tfrac{\alpha_0}{c_k+\eta},\, 1\bigr)$, winning $g_k$ at a price just above $\alpha_0$ for utility approaching $(c_k+\eta) - \alpha_0 > 0$, in place of her current $0$.
    Hence, at equilibrium $\alpha_0 \geq c_k + \eta$ for every $g_k \in S$, and in particular $\alpha_0 \geq c_{k^\ast} + \eta \geq c_\ell + \eta$.

    Bidder $0$ wins the $N$ free items and the $\ell$ items of $S$, all at price $\alpha_0$, so
    \[
        u_0(\alpha) = (N+\ell)(1-\alpha_0) \leq (N+\ell)(1 - c_\ell - \eta) = N - (N+\ell)\eta < N.
    \]
    But deviating to $\alpha_0' = 0$ wins all $N$ free items for free (whatever the competitors bid), so $u_0(0, \alpha_{-0}) \geq N > u_0(\alpha)$, contradicting that $\alpha$ is a PNE. Therefore $S = \varnothing$ at any PNE.
\end{proof}

It remains to compute the Price of Stability of $\mathcal{I}_{N,M,\eta}$ and show that it approaches $\frac{e}{e-1}$.

    By \cref{lem:pos-existence} a PNE exists, so $\mathrm{PoS}(\mathcal{I}_{N,M,\eta})$ is well-defined. By \cref{lem:pos-concede}, in every PNE bidder $0$ wins all $N$ free items, and each contested item $g_k$, which only bidder $0$ and the pair $A_k, B_k$ value, is therefore won by a competitor of value $c_k + \eta$. Every PNE thus has the same welfare
    \[
        \sw = N + \sum_{k=1}^{M}(c_k + \eta) = N + M - N\bigl(H_{N+M} - H_N\bigr) + M\eta,
    \]
    using $\sum_{k=1}^{M} c_k = \sum_{k=1}^{M}\bigl(1 - \tfrac{N}{N+k}\bigr) = M - N(H_{N+M} - H_N)$. Consequently,
    \[
        \mathrm{PoS}(\mathcal{I}_{N,M,\eta}) = \frac{N+M}{\,N + M - N(H_{N+M} - H_N) + M\eta\,}.
    \]
    Letting $\eta \to 0^+$ and then $N, M \to \infty$ with $r \coloneq M/N$ fixed, we have $H_{N+M} - H_N \to \ln(1 + r)$, so the ratio tends to $\frac{1+r}{1 + r - \ln(1+r)}$. This is maximized at $r = e-1$, where it equals $\frac{e}{e-1}$.
\end{proof}

We have now proved both the upper and the lower bound. The main result of the section follows directly.
\begin{theorem}\label{thm:inefficiency}
    For APS games, the Price of Anarchy and the Price of Stability with respect to social welfare are both equal to $\frac{e}{e-1}$.
\end{theorem}

\begin{proof}[Proof of \cref{thm:inefficiency}]
As mentioned before we are going to use the following observations to prove our bounds.
\begin{itemize}
    \item[-] An \emph{upper} bound on the PoA transfers to the PoS. Our upper bound (\cref{thm:poa-upper}) shows that \emph{every} equilibrium has welfare at least $\left(1-\tfrac{1}{e}\right)\opt$, so in particular the best one does; equivalently, $\mathrm{PoS} \le \mathrm{PoA} \le \frac{e}{e-1}$.
    \item[-] A \emph{lower} bound on the PoS transfers to the PoA. Our lower bound (\cref{thm:pos-lower}) exhibits instances in which \emph{every} equilibrium, and hence in particular the worst, is $\frac{e}{e-1}$-inefficient, so $\mathrm{PoA} \ge \mathrm{PoS} \ge \frac{e}{e-1}$.
\end{itemize}

    By \cref{thm:poa-upper} we have $\mathrm{PoA} \le \frac{e}{e-1}$, and hence $\mathrm{PoS} \le \mathrm{PoA} \le \frac{e}{e-1}$. Conversely, \cref{thm:pos-lower} gives $\mathrm{PoS} \ge \frac{e}{e-1}$, and hence $\mathrm{PoA} \ge \mathrm{PoS} \ge \frac{e}{e-1}$. Therefore $\mathrm{PoA} = \mathrm{PoS} = \frac{e}{e-1}$.
\end{proof}
\section{Future directions}

In this paper we studied APS games, in which utility-maximizing bidders each choose a single pacing multiplier in simultaneous first-price auctions. We showed that pure Nash equilibria need not exist, and in fact that no $\varepsilon$-approximate pure equilibrium is guaranteed even for a linear $\varepsilon$. We characterized the inefficiency of equilibria when they do exist, proving that both the Price of Anarchy and the Price of Stability are exactly $e/(e-1)$.
On the computational side, we established that deciding the existence of an $\varepsilon$-approximate pure equilibrium is $\mathsf{NP}$-complete, while giving polynomial-time algorithms when either the number of bidders or the number of items is constant.

A natural next step is to introduce budgets. Pacing multipliers were originally proposed as a mechanism for budget management, so the budgeted variant of APS games is arguably the setting of greatest practical interest. Our hardness construction carries over directly, since it can be embedded as an instance in which no budget constraint binds, so deciding equilibrium existence remains $\mathsf{NP}$-hard. Obtaining positive results appears substantially more difficult in this regime: budget constraints couple a bidder's behaviour across items and break the per-item reasoning that our efficiency and algorithmic arguments rely on, so even bounding the inefficiency of equilibria seems to require new techniques.

A second direction is to move beyond complete information to a Bayesian setting, in which each bidder's value is drawn from a publicly known prior but its realization is private. Here the solution concept becomes a Bayes-Nash equilibrium in pacing strategies, and the relevant questions, existence, inefficiency, and computation, must be revisited when bidders best respond in expectation over their opponents' values rather than against a known profile.

\paragraph{Use of AI tools.}
We used the large language models Claude Opus and Claude Fable in the following parts of this work. For the NP-hardness result, they helped us decide some technical parts of the proof, including the choice of the problem to reduce from. For the polynomial-time algorithm for a constant number of items, they suggested a cleaner argument for enumerating the configurations. For the bounds on the Price of Anarchy and the Price of Stability, we provided them with our initial proofs, together with related papers that use the relevant techniques, and they helped produce the tightened bounds. Finally, after we had written a first draft, we used them to identify possible errors and typos, and to suggest improvements to the language and places where technical clarifications were needed. All results were independently checked, substantially rewritten, and verified by the authors. The authors take full responsibility for all claims, proofs, references, and the final text.

\bibliographystyle{plainnat}
\bibliography{refs}

\end{document}